\documentclass[superscriptaddress, reprint, aps, prx]{revtex4-2}

\pdfoutput=1

\usepackage{amsmath}
\usepackage{amssymb}
\usepackage{amsthm}
\usepackage{mathtools}
\usepackage{natbib}
\usepackage{url}
\usepackage{textcase}
\usepackage{bm}
\usepackage{float}
\usepackage[T1]{fontenc}
\usepackage{graphicx}
\usepackage{dcolumn}
\usepackage{hyperref}
\usepackage{xcolor}
\usepackage{soul}

\graphicspath{{./images/}}
\hypersetup{
    colorlinks,
    linkcolor={blue!90!black},
    citecolor={blue!90!black},
    urlcolor={blue!90!black}
}

\newtheorem{corollary}{Corollary}
\newtheorem{result}{Result}

\makeatletter
\newtheorem*{rep@theorem}{\rep@title}
\newcommand{\newreptheorem}[2]{%
    \newenvironment{rep#1}[1]{%
        \def\rep@title{#2 \ref*{##1}}%
        \begin{rep@theorem}}%
    {\end{rep@theorem}}}
\makeatother

\usepackage{thmtools}
\usepackage{thm-restate}

\newcommand{\ket}[1]{| #1 \rangle}
\newcommand{\bra}[1]{\langle #1 |}
\newcommand{\braket}[2]{\langle #1 | #2 \rangle}
\newcommand{\ketbra}[2]{| #1 \rangle \langle #2 |}

\newcommand{\fv}[3]{\langle #1 | #2 | #3 \rangle}

\DeclareMathOperator{\End}{End}
\DeclareMathOperator{\fec}{vec}
\DeclareMathOperator{\tr}{tr}
\DeclareMathOperator{\Wg}{Wg}

\let\oldl\left
\let\oldr\right
\renewcommand{\left}{\mathopen{}\mathclose\bgroup\oldl}
\renewcommand{\right}{\aftergroup\egroup\oldr}

\begin{document}

\title{Correlations in Disordered Solvable Tensor Network States}

\author{Daniel Haag}

\email{daniel.haag@tum.de}

\affiliation{Department of Computer Science, Technische Universit{\"{a}}t M{\"{u}}nchen, Boltzmannstr.~3, 85748 Garching, Germany}

\affiliation{Munich Center for Quantum Science and Technology (MCQST), Schellingstr.~4, 80799 M{\"{u}}nchen, Germany}

\author{Richard M. Milbradt}

\email{r.milbradt@tum.de}

\affiliation{Department of Computer Science, Technische Universit{\"{a}}t M{\"{u}}nchen, Boltzmannstr.~3, 85748 Garching, Germany}

\affiliation{Munich Center for Quantum Science and Technology (MCQST), Schellingstr.~4, 80799 M{\"{u}}nchen, Germany}

\author{Christian B. Mendl}

\email{christian.mendl@tum.de}

\affiliation{Department of Computer Science, Technische Universit{\"{a}}t M{\"{u}}nchen, Boltzmannstr.~3, 85748 Garching, Germany}

\affiliation{Munich Center for Quantum Science and Technology (MCQST), Schellingstr.~4, 80799 M{\"{u}}nchen, Germany}

\affiliation{Institute for Advanced Study, Technische Universit{\"{a}}t M{\"{u}}nchen, Lichtenbergstr.~2~a, 85748 Garching, Germany}

\date{\today}

\begin{abstract}
    Solvable matrix product and projected entangled pair states evolved by dual and ternary-unitary quantum circuits have analytically accessible correlation functions. Here, we investigate the influence of disorder. Specifically, we compute the average behavior of a physically motivated two-point equal-time correlation function with respect to random disordered solvable tensor network states arising from the Haar measure on the unitary group. By employing the Weingarten calculus, we provide an exact analytical expression for the average of the \( k \)th moment of the correlation function. The complexity of the expression scales with \( k! \) and is independent of the complexity of the underlying tensor network state. Our result implies that the correlation function vanishes on average, while its covariance is nonzero.
\end{abstract}

\maketitle

\onecolumngrid

\begin{figure}[t]
	\centering
	\includegraphics{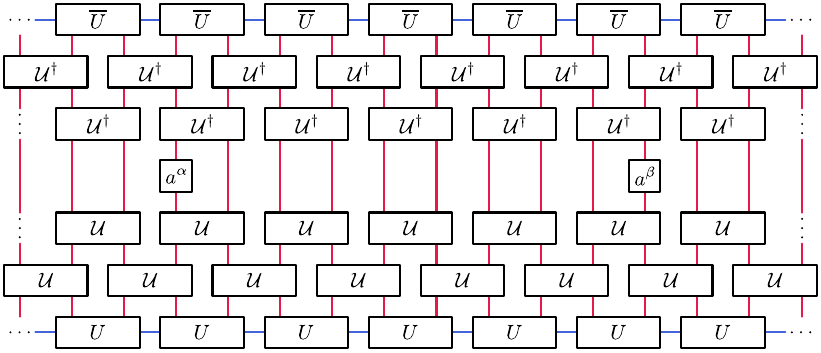}
	\caption{Tensor diagram for the two-point equal-time correlation function \( C^{\alpha \beta} (x, r, t) \) for the case of odd \( x \), \( r = 9 \), and \( t = 1 \).}
	\label{fig:invariant_mps_correlation_function}
\end{figure}

\twocolumngrid

\section{Introduction}

Correlation functions are an important subject of study in the context of quantum many-body dynamics because they encode a plethora of information about the underlying quantum many-body system. That said, it is generally hard to compute correlation functions exactly; noninteracting systems, certain integrable models~\cite{Keyserlingk2018, Chan2018}, and dual-unitary lattice models~\cite{Bertini2019} are rare exceptions.

Tensor network states~\cite{Cirac2021} represent many of the physically relevant states of a quantum-many body system. They constitute an exponentially small subset of the full Hilbert space~\cite{Poulin2011}. Their pre-eminent one-dimensional representatives, matrix product states (MPS), have been shown to faithfully represent ground states of gapped local Hamiltonians~\cite{Verstraete2006, Hastings2007, Arad2013}. With projected entangled pair states (PEPS), MPS have been generalized to two (or more) spatial dimensions. While only a weaker link between local Hamiltonians and PEPS has been proven rigorously, two-dimensional PEPS are known to efficiently represent a wide class of strongly correlated states~\cite{Cirac2021, PerezGarcia2008}. Moreover, tensor network states can be used for numerically studying the dynamics of quantum-many body systems~\cite{Daley2004, White2004, Vidal2007, Orus2008, Banuls2009, Schollwoeck2011}. However, already in one spatial dimension, the utility of MPS is typically limited by the linear growth of the entanglement entropy with time~\cite{Calabrese2005}.

In this context, dual-unitary quantum circuits~\cite{Bertini2019} and solvable MPS~\cite{Piroli2020} have exceptional features. The former describe the dynamics of a particular quantum lattice model. The defining two-particle gates are unitary in space and time. Not only do dual-unitary quantum circuits have analytically accessible correlation functions~\cite{Bertini2019, Piroli2020}, but also a range of other aspects of their dynamics can be computed exactly~\cite{Bertini2019a, Gopalakrishnan2019, Claeys2020, Bertini2020, Bertini2020a}. Solvable MPS constitute the complete class of initial states with analytically accessible dynamics~\cite{Piroli2020}. Dual-unitary quantum circuits have been realized experimentally~\cite{Mi2021, Chertkov2022}.

They have furthermore been extended to more general cases~\cite{Prosen2021, Jonay2021, Mestyan2022, Kos2023}. In particular, the concepts of dual-unitary quantum circuits and solvable MPS have been generalized to two spatial dimensions~\cite{Milbradt2023}. Ternary-unitary gates are two-times-two-particle gates that are unitary in both spatial directions and in time. Much like their one-dimensional counterparts, they have analytically accessible correlation functions. Solvable PEPS constitute a class of initial states with analytically accessible dynamics that is albeit not complete~\cite{Milbradt2023}.

Here, we investigate the average behavior of correlations in solvable tensor network states as quantified by a physically motivated two-point equal-time correlation function. Instead of introducing randomness on the level of the dual-unitary quantum circuit~\cite{Kasim2023}, we average the correlation function with respect to ensembles of random disordered solvable tensor network states arising from the Haar measure on the unitary group.

The importance of defining ensembles of random tensor network states for the purpose of exploring typical properties of physically relevant states has been recognized more than a decade ago~\cite{Garnerone2010}. MPS ensembles have been utilized to gain insights into, among other things, the typicality of expectation values of local observables~\cite{Garnerone2010, Garnerone2010a}, equilibration under Hamiltonian time evolution~\cite{Haferkamp2021}, the entropy of subsystems~\cite{Collins2013}, nonstabilizerness~\cite{Chen2022}, and the behavior of correlations~\cite{Movassagh2017, Movassagh2022, Lancien2021, Bensa2023, Svetlichnyy2022, Haag2023}.

We provide an exact analytical expression for the average of the \( k \)th moment of the two-point equal-time correlation function for ensembles of random disordered solvable MPS and PEPS. The complexity of the expression scales with \( k! \) and is independent of the complexity of the underlying tensor network state. It turns out that the correlation function vanishes on average, while its covariance is nonzero.

\section{Preliminaries} \label{sec:preliminaries}

\subsection{Dual-unitary quantum circuits} \label{sec:dual}

Dual-unitary quantum circuits were first introduced in Ref.~\cite{Bertini2019}. A dual-unitary matrix \begin{align}
    \vcenter{\hbox{\includegraphics{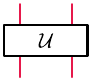}}} \in \End \left ( \mathbb{C}^d \otimes \mathbb{C}^d \right ) 
\end{align} is a unitary matrix that is unitary in space and time. That is, it satisfies \begin{align} \label{eq:temporal_unitarity}
    \vcenter{\hbox{\includegraphics{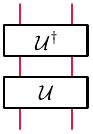}}} = \vcenter{\hbox{\includegraphics{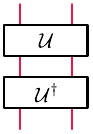}}} = \vcenter{\hbox{\includegraphics{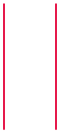}}}
\end{align} as well as \begin{subequations} \label{eq:spatial_unitarity}
    \begin{align}
        \vcenter{\hbox{\includegraphics{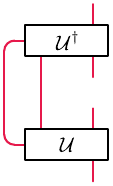}}} \, & = \, \vcenter{\hbox{\includegraphics{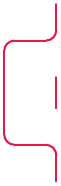}}}
    \end{align} and \begin{align}
        \vcenter{\hbox{\includegraphics{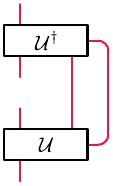}}} \, & = \, \vcenter{\hbox{\includegraphics{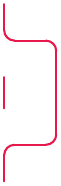}}}.
    \end{align}
\end{subequations} For \( d = 2 \), any dual-unitary matrix can be written as~\cite{Bertini2019}
\begin{align}
    \mathcal{U} = \exp (i \phi) \left ( u_+ \otimes u_- \right ) V (J) \left ( v_- \otimes v_+ \right ),
\end{align} where \( \phi, J \in \mathbb{R} \), \( u_\pm, v_\pm \in \mathrm{SU} (2) \), and \begin{align}
    V (J) = \exp \left [ - i \left ( \frac{\pi}{4} \sigma_x \otimes \sigma_x + \frac{\pi}{4} \sigma_y \otimes \sigma_y + J \sigma_z \otimes \sigma_z \right ) \right ].
\end{align} No such general parameterizations have been reported for \( d \geq 3 \).

The authors of Ref.~\cite{Bertini2019} consider the discrete and local time evolution of a one-dimensional \( 2 n \)-particle quantum system with periodic boundary conditions and local dimension \( d \). The evolution is governed by a dual-unitary matrix; one time step is realized via \begin{align} \label{eq:mps_time_evolution}
    \vcenter{\hbox{\includegraphics{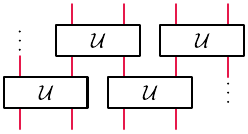}}} \in \End \left ( \left ( \mathbb{C}^d \right )^{\otimes 2 n} \right ).
\end{align} The authors find that the system has analytically accessible correlation functions~\cite{Bertini2019}.

\subsection{Solvable matrix product states} \label{sec:invariant_mps}

Like the authors of Ref.~\cite{Bertini2019}, we consider a \( 2 n \)-particle system with periodic boundary conditions and local dimension \( d \). We take the thermodynamic limit of \( n \to \infty \). In this setting, the time evolution under Eq.~\eqref{eq:mps_time_evolution} becomes intractable for general initial states. Therefore, we consider systems that can be parameterized by solvable matrix product states, which are a class of states whose dynamics under dual-unitary quantum circuits can be computed exactly~\cite{Piroli2020}. Before defining this special class of states, we briefly introduce matrix product states (MPS), shift-invariant MPS, infinite MPS, and the concept of injectivity in the context of translation-invariant MPS.

MPS are the pre-eminent tensor network structure in one dimension~\cite{Cirac2021}. A \( 2 n \)-particle MPS with periodic boundary conditions and local (physical) dimension \( d \) is given by \begin{align}
	\ket{\psi} = \sum_{i_1, \dots, i_{2 n}} \tr \left ( A_{i_1}^{(1)} \cdots A_{i_{2 n}}^{(2 n)} \right ) \ket{i_1 \cdots i_{2 n}},
\end{align} where \( A_{i_j}^{(j)} \in \End \left ( \mathbb{C}^D \right ) \) with \( i_j \in \{ 1, 2, \dots, d \} \) and \( j \in \{ 1, \dots, n \} \). \( D \) is called the bond dimension of the MPS.

Shift-invariant MPS are invariant under translations by a certain number of sites. The underlying symmetry plays a significant role in the context of this work. In the context of this section, we focus on two-site shift-invariant MPS; they are given by \begin{align} \label{eq:invariant_mps_ab}
	\ket{\psi} = \sum_{i_1, \dots, i_{2 n}} \tr \left ( A_{i_1} B_{i_2} \cdots A_{i_{2 n - 1}} B_{i_{2 n}} \right ) \ket{i_1 \cdots i_{2 n}},
\end{align} where \( A_{i_j}, B_{i_j} \in \End \left ( \mathbb{C}^D \right ) \) with \( i_j \in \{ 1, 2, \dots, d \} \) and \( j \in \{ 1, \dots, n \} \). One can define a matrix \( W \in \End \left ( \mathbb{C}^d \otimes \mathbb{C}^D \right ) \) such that \begin{align}
	\ket{\psi} & = \sum_{i_1, \dots, i_{2 n}} \tr \left ( W_{i_1, i_2} \cdots W_{i_{2 n - 1}, i_{2 n}} \right ) \ket{i_1 \cdots i_{2 n}} \label{eq:invariant_mps_w} \\
	& = \raisebox{-7pt}{\includegraphics{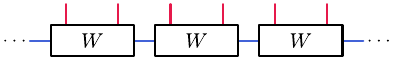}}. \label{eq:invariant_mps_w_graphical}
\end{align} In the second line, we have introduced a graphical notation in which vertical (red) legs represent physical space indices \( \left ( \mathbb{C}^d \right ) \) and horizontal (blue) legs represent bond space indices \( \left ( \mathbb{C}^D \right ) \). The periodic boundary conditions are implicit.

A two-site shift-invariant MPS thus corresponds to a translation-invariant MPS with physical dimension \( d^2 \). As we consider the thermodynamic limit of \( n \to \infty \), we are in the realm of infinite MPS~\cite{Cirac2021}.

Finally, let us state the definition of injectivity~\cite{PerezGarcia2007} in the context of translation-invariant MPS. An MPS \( \ket{\psi} \) as defined in Eq.~\eqref{eq:invariant_mps_w} is called injective if the linear map \begin{align}
    X \mapsto \sum_{i_1, \dots, i_{2 n}} \tr \left ( X W_{i_1, i_2} \cdots W_{i_{2 n - 1}, i_{2 n}} \right ) \ket{i_1 \cdots i_{2 n}}
\end{align} is injective. Importantly, injectivity is a generic property~\cite{PerezGarcia2007}.

With that, we have laid the groundwork to define solvable MPS as a class of injective two-site shift-invariant MPS that are parameterized by a unitary matrix~\cite{Piroli2020} \begin{align}
    \raisebox{-7pt}{\includegraphics{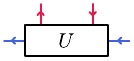}} \in \mathrm{U} (d D),
\end{align} where the arrows denote the input and output. That is, any solvable MPS can be parameterized as \begin{align}
    \ket{\psi} & = \frac{1}{\sqrt{d^n}} \sum_{i_1, \dots, i_{2 n}} \tr \left ( U_{i_1, i_2} \cdots U_{i_{2 n - 1}, i_{2 n}} \right ) \ket{i_1 \cdots i_{2 n}} \label{eq:invariant_mps} \\
    & = \frac{1}{\sqrt{d^n}} \raisebox{-7pt}{\includegraphics{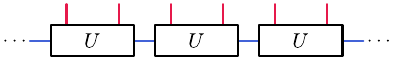}},
\end{align} where the factor \( 1 / \sqrt{d^n} \) ensures that \( \braket{\psi}{\psi} = 1 \), as we discuss in App.~\ref{app:invariant_mps_points}.

Solvable MPS facilitate the analytic computation of correlations~\cite{Piroli2020}. This is thanks to the fact that the transfer matrix of \( \ket{\psi} \),
\begin{align} \label{eq:invariant_mps_transfer_matrix}
    E = \frac{1}{d} \sum_{i_1, i_2 = 1}^d U_{i_1, i_2} \otimes \overline{U_{i_1, i_2}} = \frac{1}{d} \vcenter{\hbox{\includegraphics{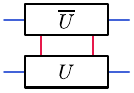}}},
\end{align} has unique left and right fixed points. As we discuss in App.~\ref{app:invariant_mps_points}, the former is given by \begin{subequations} \label{eq:mps_fixed_point}
    \begin{align}
        \bra{L} = \frac{1}{\sqrt{D}} \sum_{i = 1}^D \bra{i} \otimes \bra{i} = \frac{1}{\sqrt{D}} \vcenter{\hbox{\includegraphics{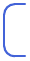}}},
    \end{align} and the latter is given by \begin{align}
        \ket{R} = \frac{1}{\sqrt{D}} \sum_{i = 1}^D \ket{i} \otimes \ket{i} = \frac{1}{\sqrt{D}} \vcenter{\hbox{\includegraphics{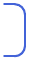}}}.
    \end{align}
\end{subequations}

In this work, we are interested in correlations quantified by the two-point equal-time correlation function \begin{align} \label{eq:mps_correlation_function}
    C^{\alpha \beta} (x, r, t) = \fv{\psi (t)}{a_x^\alpha a_{x + r}^\beta}{\psi (t)},
\end{align} where \( \ket{\psi (t)} \) is a solvable MPS evolved until time \( t \in \mathbb{N} \) by a dual-unitary quantum circuit. \( \left \{ a_x^\alpha \right \} \) with \( \alpha \in \{ 0, \dots, d^2 - 1 \} \) is a basis of the space of operators acting on site \( x \). We assume the basis to be Hilbert-Schmidt orthonormal and choose \( a^0 = I \), implying that \( \tr \left ( a^\alpha \right ) = 0 \) for \( \alpha \neq 0 \).

\( C^{\alpha \beta} (x, r, t) \) is normalized so that \( C^{0 0} (x, r, t) = 1 \). \( C^{\alpha \beta} (x, r, t) \) vanishes under a number of circumstances, as laid out in Ref.~\cite{Piroli2020}. In particular, for \( \alpha \neq 0 \) and \( \beta \neq 0 \), \begin{align} \label{eq:adrian}
    C^{\alpha \beta} (x, r, t) = \begin{cases}
        0 & x \ \text{even} \\
        0 & r \ \text{even} \\
        0 & r < 4 t + 1 \\
        D_1^{\alpha \beta} (x, r, t) & r = 4 t + 1 \\
        D_2^{\alpha \beta} (x, r, t) & r > 4 t + 1
    \end{cases},
\end{align} where our indexing convention is sketched in Fig.~\ref{fig:indexing}~(a). \( D_1^{\alpha \beta} (x, r, t) \) and \( D_2^{\alpha \beta} (x, r, t) \) are functions given by simple tensor diagrams.

As an illustrative example for the latter, let us consider the case of odd \( x \), \( r = 9 \), and \( t = 1 \) sketched in Fig.~\ref{fig:invariant_mps_correlation_function}. By exploiting the solvability of \( \ket{\psi} \) and the dual unitarity of the temporal evolution, \begin{align} \label{eq:fabian}
    & D_2^{\alpha \beta} (x, r, t) \nonumber \\
    & \qquad = \frac{\hbox{\includegraphics{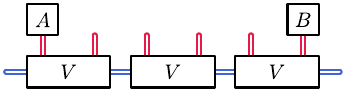}}}{d^3 D},
\end{align} where we have defined \( A = \mathcal{M}_+^{2 t} \left ( a^\alpha \right ) \), \( B = \mathcal{M}_-^{2 t} \left ( a^\beta \right ) \), and \( V = U \otimes \overline{U} \) to adopt a compact folded notation. As in Ref.~\cite{Bertini2019}, the linear maps \( \mathcal{M}_+ \) and \( \mathcal{M}_- \) are, respectively, given by \begin{subequations} \label{eq:m}
    \begin{align} 
        \mathcal{M}_+ (a) = \frac{1}{d} \tr_1 \left [ \mathcal{U}^\dagger (a \otimes I) \mathcal{U} \right ] = \frac{1}{d} \vcenter{\hbox{\includegraphics{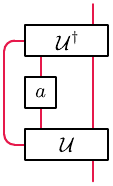}}}
    \end{align} and \begin{align}
        \mathcal{M}_- (a) = \frac{1}{d} \tr_2 \left [ \mathcal{U}^\dagger (I \otimes a) \mathcal{U} \right ] = \frac{1}{d} \vcenter{\hbox{\includegraphics{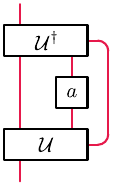}}}.
    \end{align}
\end{subequations} For more details, see App.~\ref{app:invariant_mps_correlation_function}.

We will compute the average of the \( k \)th moment of \( C^{\alpha \beta} (x, r, t) \) with respect to a measure of random disordered solvable MPS, which we will define later. The underlying idea is to draw the corresponding unitary matrices independently from the Haar measure on the unitary group. 

\subsection{Ternary-unitary quantum circuits} \label{sec:ternary}

First introduced in Ref.~\cite{Milbradt2023}, ternary-unitary quantum circuits are the generalization of dual-unitary quantum circuits to two dimensions. Just like dual-unitary matrices, ternary-unitary matrices are unitary in both space and time. The defining difference is that ternary-unitary matrices are unitary also in the additional spatial direction \( x_2 \). In a top-down perspective, we denote a ternary-unitary matrix by \begin{align}
    \vcenter{\hbox{\includegraphics{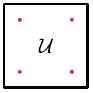}}} \in \End \left ( \mathbb{C}^d \otimes \mathbb{C}^d \otimes \mathbb{C}^d \otimes \mathbb{C}^d \right ).
\end{align}

The authors of Ref.~\cite{Milbradt2023} consider a system that is confined to a \( 2 n \times m \) lattice with periodic boundary conditions whose evolution is governed by a ternary-unitary matrix; one time step is realized via \begin{align}
    \vcenter{\hbox{\includegraphics{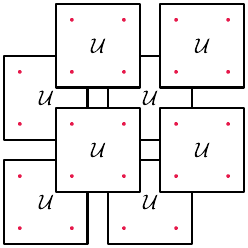}}} \in \End \left ( \left ( \mathbb{C}^d \right )^{\otimes 2 m n} \right ).
\end{align} Just like in one dimension, it turns out that the system has analytically accessible correlation functions~\cite{Milbradt2023}.

\subsection{Solvable projected entangled pair states} \label{sec:invariant_peps}

\begin{figure}[t]
	\centering
	\includegraphics{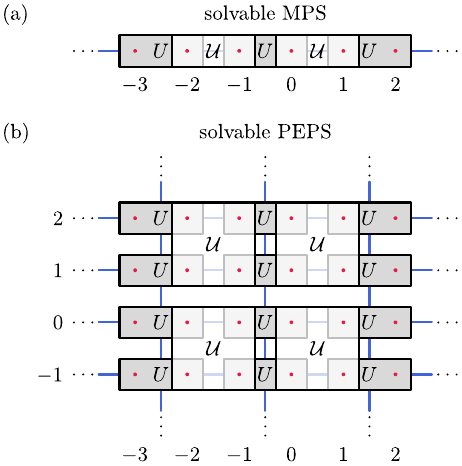}
	\caption{Top-down perspective of a solvable MPS (a) and a solvable PEPS (b) to visualize our indexing convention. The boxes labeled with \( \mathcal{U} \) are dual-unitary matrices (a) and ternary-unitary matrices (b), corresponding to the evolution by half a time step.}
	\label{fig:indexing}
\end{figure}

Projected entangled pair states (PEPS) are the generalization of MPS to two (or more) dimensions~\cite{Cirac2021}. We consider a two-dimensional system that is confined to a \( 2 n \times m \) lattice with periodic boundary conditions. We take the thermodynamic limits of \( n \to \infty \) and \( m \to \infty \). In the context of this section, we focus on a system that is two-site shift-invariant in the \( x_1 \)-direction and translation-invariant in the \( x_2 \)-direction. In analogy to Eq.~\eqref{eq:invariant_mps_w_graphical}, we denote a state of this system by \begin{align}
    \ket{\psi} = \vcenter{\hbox{\includegraphics{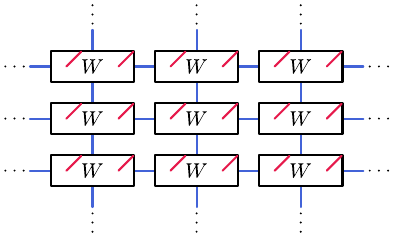}}}.
\end{align}

Solvable PEPS are a generalization of solvable MPS to two dimensions~\cite{Milbradt2023} that uses the framework of matrix product unitaries (MPU)~\cite{Cirac2017}. A PEPS \( \ket{\psi} \) is defined to be solvable if and only if there exists an MPU-generating tensor \( U \) such that \begin{align} \label{eq:invariant_peps}
    \ket{\psi} = \frac{\hbox{\includegraphics{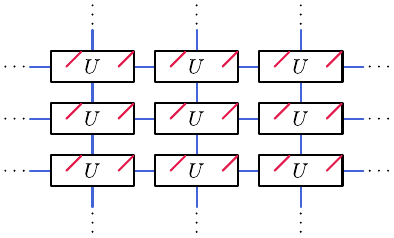}}}{\sqrt{d^{m n} D^n}},
\end{align} where the factor \( 1 / \sqrt{d^{m n} D^n} \) ensures that \( \braket{\psi}{\psi} = 1 \). In the present context, \( U \) is said to generate an MPU if the matrix product operator \begin{align}
    \vcenter{\hbox{\includegraphics{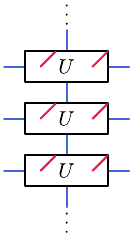}}}
\end{align} is unitary after grouping the physical legs with the horizontal bond legs: \begin{align}
    \vcenter{\hbox{\includegraphics{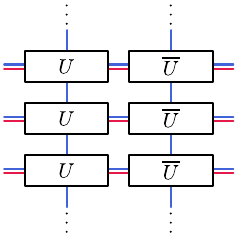}}} = \vcenter{\hbox{\includegraphics{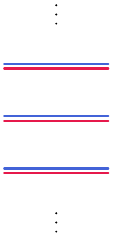}}}
\end{align} With \( V = U \otimes \overline{U} \), the condition reads \begin{align} \label{eq:mpu_condition}
    \vcenter{\hbox{\includegraphics{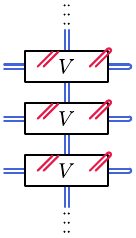}}} = \vcenter{\hbox{\includegraphics{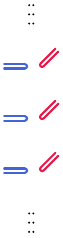}}}.
\end{align}

We shall focus on a class of solvable PEPS that is parameterized by unitary matrices \( \widetilde{U}, \widehat{U} \in \mathrm{U} (d D) \) in that \( U \) is given by
\begin{align} \label{eq:invariant_peps_parameterization}
    \vcenter{\hbox{\includegraphics{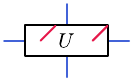}}} = \vcenter{\hbox{\includegraphics{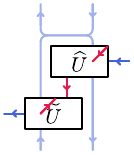}}},
\end{align} where the vertical (light blue) legs represent a bond space of dimension \( \sqrt{\mathbb{C}^D} \). This parameterization implies that \( U \) generates a simple MPU~\cite{Cirac2017, Milbradt2023}. With \( \widetilde{V} = \widetilde{U} \otimes \overline{\widetilde{U}} \) and \( \widehat{V} = \widehat{U} \otimes \overline{\widehat{U}} \), the conditions for simplicity read \begin{subequations} \label{eq:invariant_peps_simplicity}
    \begin{align} \label{eq:invariant_peps_simplicity_1}
        \vcenter{\hbox{\includegraphics{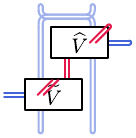}}} = \vcenter{\hbox{\includegraphics{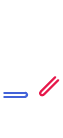}}}
    \end{align} and \begin{align} \label{eq:invariant_peps_simplicity_2}
        \vcenter{\hbox{\includegraphics{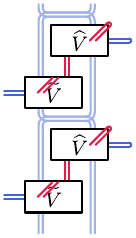}}} = \vcenter{\hbox{\includegraphics{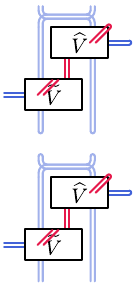}}}.
    \end{align}
\end{subequations}

As for solvable MPS, we are interested in correlations quantified by the two-point equal-time correlation function \begin{align} \label{eq:peps_correlation_function}
    C^{\alpha \beta} (x, r, t) = \fv{\psi (t)}{a_x^\alpha a_{x + r}^\beta}{\psi (t)},
\end{align} where \( \ket{\psi (t)} \) is a solvable PEPS evolved until time \( t \in \mathbb{N} \) by a ternary-unitary quantum circuit, \( \left \{ a_x^\alpha \right \} \) with \( \alpha \in \{ 0, \dots, d^2 - 1 \} \) is a basis of the space of operators acting on site \( x \), and \( C^{\alpha \beta} (x, r, t) \) is normalized so that \( C^{0 0} (x, r, t) = 1 \). Without loss of generality, we shall assume that \( r_2 \geq 0 \)~\cite{Milbradt2023}. As laid out in Ref.~\cite{Milbradt2023}, for \( \alpha \neq 0 \) and \( \beta \neq 0 \), \begin{align}
    C^{\alpha \beta} (x, r, t) & = \Theta \left ( 4 t + x_2 \bmod{2} + 1 - r_2 \right ) \nonumber \\
    & \hphantom{{} = {}} {} \times \begin{cases}
        0 & x_1 \ \text{even} \\
        0 & r_1 \ \text{even} \\
        0 & r_1 < 4 t + 1 \\
        D_1^{\alpha \beta} (x, r, t) & r_1 = 4 t + 1 \\
        D_2^{\alpha \beta} (x, r, t) & r_1 > 4 t + 1
    \end{cases}, \label{eq:simone}
\end{align} where our indexing convention is sketched in Fig~\ref{fig:indexing}~(b), and \( \Theta (\cdot) \) denotes the Heaviside step function. \( D_1^{\alpha \beta} (x, r, t) \) and \( D_2^{\alpha \beta} (x, r, t) \) are once again functions given by simple tensor diagrams.

In analogy to the one-dimensional case [see Eq.~\eqref{eq:fabian}], let us consider the case of odd \( x_1 \), even \( x_2 \), \( r_1 = 9 \), \( r_2 = 0 \), and \( t = 1 \) as an example for \( D_2^{\alpha \beta} (x, r, t) \). In the same folded notation, \begin{align}
    & D_2^{\alpha \beta} (x, r, t) \nonumber \\
    & \qquad = \frac{\hbox{\includegraphics{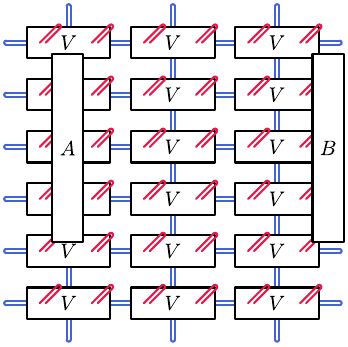}}}{d^{18} D^9}, \label{eq:invariant_peps_correlation_function}
\end{align} where \( A \) and \( B \), respectively, arise from \( a^\alpha \) and \( a^\beta \) through the application of linear maps similar to those defined in Eq.~\eqref{eq:m}~\cite{Milbradt2023}. For more details and a representation of \( D_2^{\alpha \beta} (x, r, t) \) in a nonfolded notation, see Ref.~\cite{Milbradt2023}. See also App.~\ref{app:peps_correlation_function}, where we discuss the more general disordered case.

As in one dimension, we will compute the average of the \( k \)th moment of \( C^{\alpha \beta} (x, r, t) \) with respect to a measure of random disordered solvable PEPS. The underlying idea is to draw the corresponding unitary matrices independently from the Haar measure on the unitary group.

\subsection{\texorpdfstring{\( k \)}{k}-fold twirl}
\label{sec:k-twirl}

As alluded to in Secs.~\ref{sec:invariant_mps} and \ref{sec:invariant_peps}, we will compute averages with respect to measures arising from the Haar measure on the unitary group. We achieve this by employing the \( k \)-fold twirl, which we define in this section.

Let \( X \in \End \bigl ( \left ( \mathbb{C}^q \right )^{\otimes k} \bigr ) \). The \( k \)-fold twirl of \( X \) with respect to the Haar measure on the unitary group \( \mathrm{U} (q) \) is defined~\cite{Collins2006, Roberts2017, Brandao2021} as \begin{align} \label{eq:twirl_before}
	\mathcal{T}_\mathrm{U}^{(k)} (X) = \int \mathrm{d} U \, U^{\otimes k} X \left ( U^\dagger \right )^{\otimes k}.
\end{align} One can employ the Schur-Weyl duality for unitary groups to show~\cite{Collins2003, Collins2006, Haag2023} that \begin{align} \label{eq:twirl_after}
	\mathcal{T}_\mathrm{U}^{(k)} (X) = \sum_{\sigma, \tau \in S_k} \Wg \left ( \sigma \tau^{- 1}, q \right ) P_\sigma^{(q)} \tr \left [ X \left ( P_\tau^{(q)} \right )^T \right ],
\end{align} where \begin{align}
	P_\pi^{(q)} : v_1 \otimes \cdots \otimes v_k \mapsto v_{\sigma^{- 1} (1)} \otimes \cdots \otimes v_{\sigma^{- 1} (k)}
\end{align} is the representation of \( \pi \in S_k \) on \( \left ( \mathbb{C}^q \right )^{\otimes k} \), where \( S_k \) is the symmetric group. \( \Wg \left ( \sigma \tau^{- 1}, q \right ) = \left ( G^{- 1} \right )_{\sigma \tau} \)~\footnote{Although the Weingarten matrix \( W = G^{- 1} \) exits only if \( k \leq q \)~\cite{Collins2021}, the Weingarten function can easily be extended to \( k > q \)~\cite{Collins2006}.} is the Weingarten function, where \( G \in \End ( \mathbb{R}^{k!} ) \) denotes the Gram matrix whose entries are given by \begin{align} \label{eq:gram}
	G_{\sigma \tau} = \tr \left [ P_\sigma^{(q)} \left ( P_\tau^{(q)} \right )^T \right ] = q^{\# \left ( \sigma \tau^{- 1} \right )}.
\end{align} Here, \( \# (\pi) \) counts the number of cycles in the decomposition of \( \pi \in S_k \) into disjoint cycles. Thus, \( \Wg (\pi, q) \) depends only on the conjugacy class of \( \pi \)~\cite{Collins2003}.

\subsection{Graphical notation} \label{sec:graphical}

In this section, we introduce the graphical notation used throughout this work. It coincides with that of Ref.~\cite{Haag2023}. To keep the images compact, we employ the operator-vector correspondence~\cite{Watrous2018} (that is, the vectorization of linear operators). Let \( \{ \ket{i} \} \) denote the standard basis of \( \mathbb{C}^q \). Then, the operator-vector correspondence is defined by \begin{align}
	\fec (\ketbra{i}{j}) = \ket{i} \otimes \ket{j}
\end{align} and extended linearly to the vector space at large.

Because we consider the standard (product) basis to be fixed, we do not distinguish between tensors (as multidimensional arrays) and their basis-independent counterparts (such as vectors and operators). Let \( X \in \End \bigl ( \left ( \mathbb{C}^q \right )^{\otimes k} \bigr ) \). Using the operator-vector correspondence, we denote it by \begin{align}
	\raisebox{-7pt}{\includegraphics{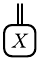}} = \fec (X).
\end{align} Note that the orientation of the legs does not have any meaning in our images. That is, \begin{align}
	\vcenter{\hbox{\includegraphics{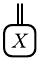}}} = \vcenter{\hbox{\includegraphics{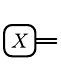}}} = \vcenter{\hbox{\includegraphics{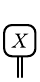}}} = \vcenter{\hbox{\includegraphics{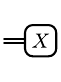}}}.
\end{align} When we need the transpose of an operator, we will explicitly use \begin{align}
	\raisebox{-7pt}{\includegraphics{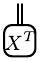}} = \fec \left ( X^T \right ).
\end{align} As such, when we contract two operators \( X \) and \( Y \), we mean the trace of their product: \begin{align}
	\vcenter{\hbox{\includegraphics{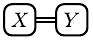}}} = \tr (X Y)
\end{align}

The most frequent operator we will encounter is \begin{align}
	\raisebox{-7pt}{\includegraphics{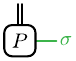}} = \fec \left ( P_{\sigma}^{(q)} \right ),
\end{align} where the horizontal (green) leg indexes permutations. The contraction of two permutation operators is given by \begin{align} \label{eq:permutation_contraction}
    \vcenter{\hbox{\includegraphics{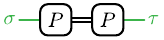}}} = \tr \left ( P_\sigma^{(q)} P_\tau^{(q)} \right ) = q^{\# (\sigma \tau)}.
\end{align}

In the following, we will not explicitly write the operator \( \fec \), as it shall be clear from the context.

With the definition of the Weingarten matrix, \begin{align}
	\vcenter{\hbox{\includegraphics{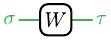}}} = \Wg \left ( \sigma \tau^{- 1}, q \right ),
\end{align} we can write the \( k \)-fold twirl [see Eq.~\eqref{eq:twirl_after}] as \begin{align} \label{eq:twirl_graphical}
	\mathcal{T}_\mathrm{U}^{(k)} (X) = \raisebox{-32pt}{\includegraphics{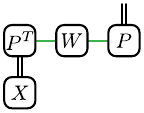}},
\end{align} where the contraction of two green legs corresponds to a summation over the permutations in \( S_k \).

\section{Correlations in disordered solvable MPS} \label{sec:yorgos}

\subsection{Disordered solvable MPS} \label{sec:mps}

We compute the average of the two-point equal-time correlation function \( C^{\alpha \beta} (x, r, t) \) [see Eq.~\eqref{eq:mps_correlation_function}] with respect to a certain measure of random disordered solvable MPS. It will prove straightforward to also compute the average of the \( k \)th moment of the correlation function.

The measure arises from two-site shift-invariant MPS [see Eq.~\eqref{eq:invariant_mps}] by allowing the unitary matrices \( U \) to be different from another. We find it necessary to retain some symmetry to prove solvability, which is why we assume a \( 2 v \)-site shift invariance, where \begin{align} \label{eq:s}
    v \geq \frac{r - 4 t - 3}{2} \equiv s
\end{align} for any \( r \) of interest. For a detailed definition of our class of disordered solvable MPS, see App.~\ref{app:mps}, and for a discussion of the necessity to retain symmetry to prove solvability, see App.~\ref{app:mps_points}. We argue that the symmetry does not constitute a limitation on full disorder because we consider the thermodynamic limit of \( n \to \infty \), leading us to adapt the term disordered solvable MPS.

The measure of random disordered solvable MPS is then defined by drawing the unitary matrices \( U^{(j)} \) with \( j \in \{ 1, \dots, v \} \) independently from the Haar measure on the unitary group \( \mathrm{U} (d D) \). This choice is motivated by the definition of random MPS as first seen in Refs.~\cite{Garnerone2010, Garnerone2010a} and more recently in Refs.~\cite{Haferkamp2021, Haag2023}.

As we discuss in App.~\ref{app:mps_correlation_function}, the procedure of simplifying \( D_1^{\alpha \beta} (x, r, t) \) and \( D_2^{\alpha \beta} (x, r, t) \) is all but identical to that of the nondisordered case. Using again the case of odd \( x \), \( r = 9 \), and \( t = 1 \) as an example for the latter, we have \begin{align}
    & D_2^{\alpha \beta} (x, r, t) \nonumber \\
    & \qquad = \frac{\hbox{\includegraphics{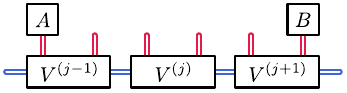}}}{d^3 D},
\end{align} where \( V^{(j)} = U^{(j)} \otimes \overline{U^{(j)}} \) with \( j \in \{ 1, \dots, v \} \). Similarly, the \( k \)th moment of \( D_2^{\alpha \beta} (x, r, t) \) is given by \begin{align} \label{eq:mps_moment}
    & \bigl [ D_2^{\alpha \beta} (x, r, t) \bigr ]^k \nonumber \\
    & \qquad = \frac{\hbox{\includegraphics{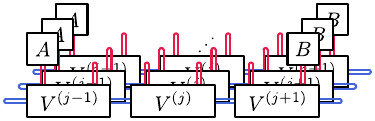}}}{\left ( d^3 D \right )^k},
\end{align} with \( k \) copies of \( D_2^{\alpha \beta} (x, r, t) \).

\subsection{Computing averages} \label{sec:mps_average}

In this section, we introduce the analytical tool that makes the computation of averages with respect to our measure of random disordered solvable MPS comparatively simple. It relies on the fact that we can draw the unitary matrices \( U^{(j)} \) with \( j \in \{ 1, \dots, v \} \) independently from the Haar measure on the unitary group \( \mathrm{U} (d D) \).

We compute the \( k \)-fold twirl [see Eq.~\eqref{eq:twirl_graphical}] independently for each pair of two sites. By doing so, we obtain
the building block \begin{align}
	\raisebox{-7pt}{\includegraphics{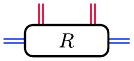}} & = \int \mathrm{d} U^{(j)} \, \raisebox{-14pt}{\includegraphics{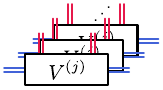}} \label{eq:mps_r_1} \\
	& = \raisebox{-7pt}{\includegraphics{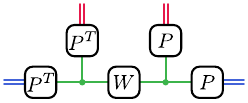}}, \label{eq:mps_r_2}
\end{align} where we have used that \begin{align}
    P_\sigma^{(d D)} = P_\sigma^{(d)} \otimes P_\sigma^{(D)}.
\end{align} Accordingly, the (green) dot represents a Kronecker delta on three permutation indices.

The average of the \( k \)th moment of \( D_2^{\alpha \beta} (x, r, t) \) is then given by \begin{align}
	& \mathbb{E} \bigl [ D_2^{\alpha \beta} (x, r, t) \bigr ]^k \nonumber \\
	& \qquad = \frac{\hbox{\includegraphics{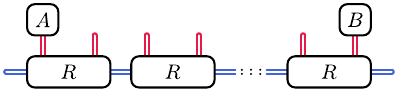}}}{\left ( d^{s + 2} D \right )^k},
\end{align} where \( s \) is defined in Eq.~\eqref{eq:s}. The factor \( 1 / \left ( d^{s + 2} D \right )^k \) ensures that \( D_2^{0 0} (x, r, t) = 1 \) for arbitrary but compatible \( r \) and \( t \).

We could, in principle, work with the building block \( R \). However, it is computationally disadvantageous to have dangling bond (blue) legs whose dimension grows with \( D \). By postponing the contraction of permutation-valued (green) legs, we obtain a building block with fixed dimension for fixed \( k \). With that building block, the average of the \( k \)th moment of \( D_2^{\alpha \beta} (x, r, t) \) is given by \begin{align}
	& \mathbb{E} \bigl [ D_2^{\alpha \beta} (x, r, t) \bigr ]^k \nonumber \\
	& \qquad = \frac{\hbox{\includegraphics{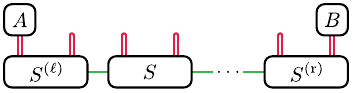}}}{\left ( d^{s + 2} D \right )^k} \\
	& \qquad \equiv \frac{\hbox{\includegraphics{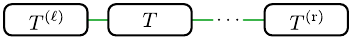}}}{\left ( d^{s + 2} D \right )^k}, \label{eq:sebastian}
\end{align} where the matrix \( T \in \mathbb{R}^{k! \times k!} \) is defined with respect to the standard basis of \( \mathbb{R}^{k!} \) by enumerating the permutations in \( S_k \).

With that, we have reduced the seemingly daunting task of computing the average of the \( k \)th moment of \( D_2^{\alpha \beta} (x, r, t) \) to evaluating the comparatively simple expression \begin{align}
    \mathbb{E} \bigl [ D_2^{\alpha \beta} (x, r, t) \bigr ]^k = \frac{1}{d^{s + 2} D} \fv{T^{(\ell)}}{T^s}{T^{(\mathrm{r})}}.
\end{align}

The definitions of \( T \), \( T^{(\ell)} \), and \( T^{(\mathrm{r})} \) are stated in App.~\ref{app:mps_average}. We furthermore provide a simple Mathematica package~\cite{GitHub} that defines \( T \), \( T^{(\ell)} \), and \( T^{(\mathrm{r})} \) for \( k \in \{ 1, \dots, 20 \} \). The package relies on the package provided by the authors of Ref.~\cite{Fukuda2019} for evaluating the Weingarten function.

Even more straightforward, the average of the \( k \)th moment of \( D_1^{\alpha \beta} (x, r, t) \) is given by \begin{align} \label{eq:mps_special_average}
	\mathbb{E} \bigl [ D_1^{\alpha \beta} (x, r, t) \bigr ]^k & = \frac{1}{\left ( d D \right )^k} \raisebox{-7pt}{\includegraphics{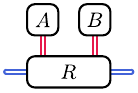}}.
\end{align} We refer to App.~\ref{app:mps_average} for more details.

\subsection{Results} \label{sec:mps_results}

We are now in the position to state our first main result. It is an immediate consequence of the previous section. 

\begin{result} \label{res:mps_average}
    The average of the \( k \)th moment of \( D_1^{\alpha \beta} (x, r, t) \) with respect to the random disordered solvable MPS ensemble is given by the tensor diagram shown in Eq.~\eqref{eq:mps_special_average}, and that of \( D_2^{\alpha \beta} (x, r, t) \) is given by \begin{align}
        \mathbb{E} \bigl [ D_2^{\alpha \beta} (x, r, t) \bigr ]^k = \frac{1}{d^{s + 2} D} \fv{T^{(\ell)}}{T^s}{T^{(\mathrm{r})}},
    \end{align} where \( s \) is defined in Eq.~\eqref{eq:s}.
\end{result}

We shall state the case of \( k = 1 \) as a corollary of Result~\ref{res:mps_average}. We prove the statement in App.~\ref{app:mps_average_correlation_function}.

\begin{restatable}{corollary}{mpsaveragecorrelationfunction} \label{cor:mps_average_correlation_function}
	The averages of \( D_1^{\alpha \beta} (x, r, t) \) and \( D_2^{\alpha \beta} (x, r, t) \) with respect to the random disordered solvable MPS ensemble are given by \begin{align}
        \mathbb{E} D_1^{\alpha \beta} (x, r, t) = \mathbb{E} D_2^{\alpha \beta} (x, r, t) = \frac{1}{d^2} \tr \left ( a^\beta \right ) \tr \left ( a^\beta \right ),
    \end{align} implying that they vanish, except for the trivial case of \( \alpha = \beta = 0 \).
\end{restatable}

Corollary~\ref{cor:mps_average_correlation_function} implies that \( \mathbb{E} \ketbra{\psi}{\psi} = I / d^n \). Given that the ensemble arises from the Haar measure on the unitary group, it is intuitive that the average disordered solvable MPS is given by the maximally mixed state~\cite{Haferkamp2021}. The Haar average of physically motivated correlation functions has been found to vanish also in different contexts~\cite{Keyserlingk2018}.

To underline the power of the Weingarten calculus, we shall also state the case of \( k = 2 \) as another corollary of Result~\ref{res:mps_average}. 

\begin{corollary}
    The average of the second moment of \( D_2^{\alpha \beta} (x, r, t) \) is given by \begin{align}
        & \mathbb{E} \bigl [ D_2^{\alpha \beta} (x, r, t) \bigr ]^2 = \frac{1}{d^4 D^2} \fv{T^{(\ell)}}{\begin{pmatrix} 1 & \displaystyle \frac{d^2 D - D}{d^2 D^2 - 1} \\[1em] 0 & \displaystyle \frac{D^2 - 1}{d^2 D^2 - 1} \end{pmatrix}^s}{T^{(\mathrm{r})}},
    \end{align} where \begin{align}
        \ket{T^{(\ell)}} = \begin{pmatrix} D^2 \tr (A)^2 \\[1em] \displaystyle \frac{\left ( d^2 - 1 \right ) D^3 \tr (A)^2 + d D \left ( D^2 - 1 \right ) \tr \left ( S A^{\otimes 2} \right )}{d^2 D^2 - 1} \end{pmatrix}
    \end{align} and \begin{align}
        \ket{T^{(\mathrm{r})}} = \begin{pmatrix} \displaystyle \frac{d^2 D^2 \tr (B)^2 - d \tr \left ( S B^{\otimes 2} \right )}{d^2 D^2 - 1} \\[1em] \displaystyle \frac{d D \tr \left ( S B^{\otimes 2} \right ) - D \tr (B)^2}{d^2 D^2 - 1} \end{pmatrix}
    \end{align} with \( S = P_{(1 2)}^{(d)} \).
\end{corollary}

We refer to the aforementioned Mathematica package~\cite{GitHub} for the analytical expressions for \( T \), \( T^{(\ell)} \), and \( T^{(\mathrm{r})} \) for \( k \geq 3 \). While their complexity scales with \( k! \), the expressions are concise and exact.

We stress that the tool underlying Result~\ref{res:mps_average} relies on the fact that we can draw the unitary matrices \( U^{(j)} \) with \( j \in \{ 1, \dots, v \} \) independently from the Haar measure on the unitary group \( \mathrm{U} (d D) \). The tool can thus not be used to study the average behavior of correlations in nondisordered solvable MPS.

\section{Correlations in disordered solvable PEPS} \label{sec:flavio}

\subsection{Disordered solvable PEPS} \label{sec:peps}

The measure with respect to which we compute the average of the two-point equal-time correlation function \( C^{\alpha \beta} (x, r, t) \) [see Eq.~\eqref{eq:peps_correlation_function}] arises from the PEPS defined by Eqs.~\eqref{eq:invariant_peps} and \eqref{eq:invariant_peps_parameterization} by allowing the unitary matrices \( \widetilde{U} \) and \( \widehat{U} \) to be different from another, similar to the one-dimensional case. Once again, we find it necessary to retain some symmetry to prove solvability, which is why we assume a \( 2 v \)-site shift invariance in the \( x_1 \)-direction, where \( v \geq s \) with \( s \) defined in Eq.~\eqref{eq:s}. For a detailed definition of our class of disordered solvable PEPS, see App.~\ref{app:peps}, and for a discussion of the necessity to retain symmetry to prove solvability, see App.~\ref{app:peps_points}.

The measure of random disordered solvable PEPS is then defined by drawing the unitary matrices \( \widetilde{U}^{(i, j)} \) and \( \widehat{U}^{(i, j)} \) with \( i \in \{ 1, \dots, v \} \) and \( j \in \{ 1, \dots, m \} \) independently from the Haar measure on the unitary group \( \mathrm{U} (d D) \).

As we discuss in App.~\ref{app:peps_correlation_function}, the procedure of simplifying \( D_1^{\alpha \beta} (x, r, t) \) and \( D_2^{\alpha \beta} (x, r, t) \) is analogous to that in one dimension.

\subsection{Computing averages} \label{peps:peps_average}

The procedure of computing averages is identical to that in one dimension. Drawing the unitary matrices \( \widetilde{U}^{(i, j)} \) and \( \widehat{U}^{(i, j)} \) with \( i \in \{ 1, \dots, v \} \) and \( j \in \{ 1, \dots, n \} \) independently from the Haar measure on the unitary group \( \mathrm{U} (d D) \), we compute the \( k \)-fold twirl independently for each pair of sites. By doing so, we obtain the building block \begin{align}
	\vcenter{\hbox{\includegraphics{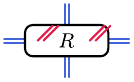}}} & = \int \mathrm{d} \widetilde{U}^{(i, j)} \mathrm{d} \widehat{U}^{(i, j)} \, \vcenter{\hbox{\includegraphics{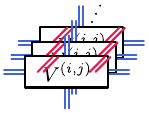}}} \label{eq:peps_r_1} \\
	& = \int \mathrm{d} \widetilde{U}^{(i, j)} \mathrm{d} \widehat{U}^{(i, j)} \, \vcenter{\hbox{\includegraphics{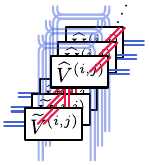}}}. \label{eq:peps_r_2}
\end{align}

As in one dimension, we could work with the building block \( R \), but we shall cut permutation-valued legs instead of bond legs. For the case of odd \( x_1 \), even \( x_2 \), \( r_1 = 9 \), \( r_2 = 0 \), and \( t = 1 \), the average of the \( k \)th moment of \( D_2^{\alpha \beta} (x, r, t) \) is then given by \begin{align}
    & \mathbb{E} \bigl [ D_2^{\alpha \beta} (x, r, t) \bigr ]^k \nonumber \\
    & \qquad = \frac{\hbox{\includegraphics{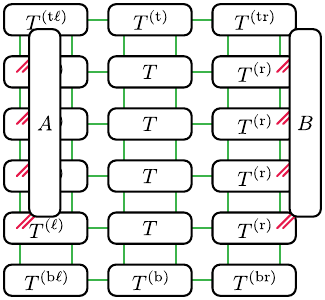}}}{\left ( d^{18} D^9 \right )^k}, \label{eq:peps_average}
\end{align} where the definitions of the tensors are stated in App.~\ref{app:peps_average}. We furthermore provide an additional Mathematica package~\cite{GitHub} that defines the tensors for \( k \in \{ 1, \dots, 20 \} \).

For the case of odd \( x_1 \), even \( x_2 \), \( r_1 = 5 \), \( r_2 = 0 \), and \( t = 1 \), the average of the \( k \)th moment of \( D_1^{\alpha \beta} (x, r, t) \) is given by \begin{align} \label{eq:peps_special_average}
    & \mathbb{E} \bigl [ D_1^{\alpha \beta} (x, r, t) \bigr ]^k = \frac{1}{\left ( d^6 D^7 \right )^k} \vcenter{\hbox{\includegraphics{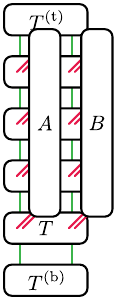}}},
\end{align} where the definitions of the tensors are also stated in App.~\ref{app:peps_average}.

\subsection{Results} \label{sec:peps_results}

We are now in the position to state our second main result. It too is an immediate consequence of the previous section.

\begin{result} \label{res:peps_average}
    The average of the \( k \)th moment of \( D_1^{\alpha \beta} (x, r, t) \) with respect to the random disordered solvable PEPS ensemble is given by a tensor diagram as shown in Eq.~\eqref{eq:peps_special_average}, and that of \( D_2^{\alpha \beta} (x, r, t) \) is given by a tensor diagram as shown in Eq.~\eqref{eq:peps_average}.
\end{result}

We shall state the case of \( k = 1 \) as a corollary of Result~\ref{res:peps_average}. We prove the statement in App.~\ref{app:peps_average_correlation_function}.

\begin{restatable}{corollary}{pepsaveragecorrelationfunction} \label{cor:peps_average_correlation_function}
	The averages of \( D_1^{\alpha \beta} (x, r, t) \) and \( D_2^{\alpha \beta} (x, r, t) \) with respect to the random disordered solvable PEPS ensemble are given by \begin{align}
        \mathbb{E} D_1^{\alpha \beta} (x, r, t) = \mathbb{E} D_2^{\alpha \beta} (x, r, t) = \frac{1}{d^2} \tr \left ( a^\beta \right ) \tr \left ( a^\beta \right ),
    \end{align} implying that they vanish, except for the trivial case of \( \alpha = \beta = 0 \).
\end{restatable}

As in one dimension, Corollary~\ref{cor:peps_average_correlation_function} implies that \( \mathbb{E} \ketbra{\psi}{\psi} = I / d^n \). Given that also the random disordered solvable PEPS ensemble arises from the Haar measure on the unitary group, this is intuitive.

We refrain from providing explicit expressions for \( k \geq 2 \) for conciseness and refer to the aforementioned Mathematica package~\cite{GitHub}.

\section{Conclusions and outlook} \label{sec:conclusion}

We have investigated the average behavior of a physically motivated two-point equal-time correlation function for ensembles of random disordered solvable MPS and PEPS. By leveraging the Weingarten calculus, we have provided an exact analytical expression for the average of its \( k \)th moment. The complexity of the expression scales with \( k! \) and is independent of the complexity of the underlying tensor network state. Our result implies that the correlation function vanishes on average, while its covariance is nonzero.

A natural extension of our work is the investigation of correlations in solvable PEPS that are more general than those defined by the particular parameterization in Eq.~\eqref{eq:invariant_peps_parameterization}. It would furthermore be interesting to study the average behavior of correlations in the nondisordered case. As we have pointed out, however, the Weingarten calculus quickly loses its power in that case.

Another intriguing possibility is the use of our analytical expressions as reference data for randomized benchmarking protocols for quantum computers. From a more general theoretical perspective, we deem the computation of quantities such as entanglement entropies and out-of-time-order correlators (OTOCs) based on the mathematical framework laid out in this work a rewarding research project for the future.

\begin{acknowledgments}
    D.H.\ thanks Georgios Styliaris and Luna Cesari for fruitful discussions. The research is part of the Munich Quantum Valley, which is supported by the Bavarian State Government with funds from the High-Tech Agenda Bavaria Plus.
\end{acknowledgments}

\bibliography{references}

\begin{thebibliography}{54}%
\makeatletter
\providecommand \@ifxundefined [1]{%
 \@ifx{#1\undefined}
}%
\providecommand \@ifnum [1]{%
 \ifnum #1\expandafter \@firstoftwo
 \else \expandafter \@secondoftwo
 \fi
}%
\providecommand \@ifx [1]{%
 \ifx #1\expandafter \@firstoftwo
 \else \expandafter \@secondoftwo
 \fi
}%
\providecommand \natexlab [1]{#1}%
\providecommand \enquote  [1]{``#1''}%
\providecommand \bibnamefont  [1]{#1}%
\providecommand \bibfnamefont [1]{#1}%
\providecommand \citenamefont [1]{#1}%
\providecommand \href@noop [0]{\@secondoftwo}%
\providecommand \href [0]{\begingroup \@sanitize@url \@href}%
\providecommand \@href[1]{\@@startlink{#1}\@@href}%
\providecommand \@@href[1]{\endgroup#1\@@endlink}%
\providecommand \@sanitize@url [0]{\catcode `\\12\catcode `\$12\catcode
  `\&12\catcode `\#12\catcode `\^12\catcode `\_12\catcode `\%12\relax}%
\providecommand \@@startlink[1]{}%
\providecommand \@@endlink[0]{}%
\providecommand \url  [0]{\begingroup\@sanitize@url \@url }%
\providecommand \@url [1]{\endgroup\@href {#1}{\urlprefix }}%
\providecommand \urlprefix  [0]{URL }%
\providecommand \Eprint [0]{\href }%
\providecommand \doibase [0]{https://doi.org/}%
\providecommand \selectlanguage [0]{\@gobble}%
\providecommand \bibinfo  [0]{\@secondoftwo}%
\providecommand \bibfield  [0]{\@secondoftwo}%
\providecommand \translation [1]{[#1]}%
\providecommand \BibitemOpen [0]{}%
\providecommand \bibitemStop [0]{}%
\providecommand \bibitemNoStop [0]{.\EOS\space}%
\providecommand \EOS [0]{\spacefactor3000\relax}%
\providecommand \BibitemShut  [1]{\csname bibitem#1\endcsname}%
\let\auto@bib@innerbib\@empty
\bibitem [{\citenamefont {von Keyserlingk}\ \emph {et~al.}(2018)\citenamefont
  {von Keyserlingk}, \citenamefont {Rakovszky}, \citenamefont {Pollmann},\ and\
  \citenamefont {Sondhi}}]{Keyserlingk2018}%
  \BibitemOpen
  \bibfield  {author} {\bibinfo {author} {\bibfnamefont {C.~W.}\ \bibnamefont
  {von Keyserlingk}}, \bibinfo {author} {\bibfnamefont {T.}~\bibnamefont
  {Rakovszky}}, \bibinfo {author} {\bibfnamefont {F.}~\bibnamefont
  {Pollmann}},\ and\ \bibinfo {author} {\bibfnamefont {S.~L.}\ \bibnamefont
  {Sondhi}},\ }\bibfield  {title} {\bibinfo {title} {Operator {H}ydrodynamics,
  {OTOCs}, and {E}ntanglement {G}rowth in {S}ystems without {C}onservation
  {L}aws},\ }\href {https://doi.org/10.1103/physrevx.8.021013} {\bibfield
  {journal} {\bibinfo  {journal} {Phys. Rev. X}\ }\textbf {\bibinfo {volume}
  {8}},\ \bibinfo {pages} {021013} (\bibinfo {year} {2018})},\ \Eprint
  {https://arxiv.org/abs/1705.08910} {arXiv:1705.08910 [cond-mat.str-el]}
  \BibitemShut {NoStop}%
\bibitem [{\citenamefont {Chan}\ \emph {et~al.}(2018)\citenamefont {Chan},
  \citenamefont {Luca},\ and\ \citenamefont {Chalker}}]{Chan2018}%
  \BibitemOpen
  \bibfield  {author} {\bibinfo {author} {\bibfnamefont {A.}~\bibnamefont
  {Chan}}, \bibinfo {author} {\bibfnamefont {A.~D.}\ \bibnamefont {Luca}},\
  and\ \bibinfo {author} {\bibfnamefont {J.~T.}\ \bibnamefont {Chalker}},\
  }\bibfield  {title} {\bibinfo {title} {Solution of a {M}inimal {M}odel for
  {M}any-{B}ody {Q}uantum {C}haos},\ }\href
  {https://doi.org/10.1103/physrevx.8.041019} {\bibfield  {journal} {\bibinfo
  {journal} {Phys. Rev. X}\ }\textbf {\bibinfo {volume} {8}},\ \bibinfo {pages}
  {041019} (\bibinfo {year} {2018})},\ \Eprint
  {https://arxiv.org/abs/1712.06836} {arXiv:1712.06836 [cond-mat.stat-mech]}
  \BibitemShut {NoStop}%
\bibitem [{\citenamefont {Bertini}\ \emph
  {et~al.}(2019{\natexlab{a}})\citenamefont {Bertini}, \citenamefont {Kos},\
  and\ \citenamefont {Prosen}}]{Bertini2019}%
  \BibitemOpen
  \bibfield  {author} {\bibinfo {author} {\bibfnamefont {B.}~\bibnamefont
  {Bertini}}, \bibinfo {author} {\bibfnamefont {P.}~\bibnamefont {Kos}},\ and\
  \bibinfo {author} {\bibfnamefont {T.}~\bibnamefont {Prosen}},\ }\bibfield
  {title} {\bibinfo {title} {Exact {C}orrelation {F}unctions for
  {D}ual-{U}nitary {L}attice {M}odels in $1 + 1$ {D}imensions},\ }\href
  {https://doi.org/10.1103/physrevlett.123.210601} {\bibfield  {journal}
  {\bibinfo  {journal} {Phys. Rev. Lett.}\ }\textbf {\bibinfo {volume} {123}},\
  \bibinfo {pages} {210601} (\bibinfo {year} {2019}{\natexlab{a}})},\ \Eprint
  {https://arxiv.org/abs/1904.02140} {arXiv:1904.02140 [cond-mat.stat-mech]}
  \BibitemShut {NoStop}%
\bibitem [{\citenamefont {Cirac}\ \emph {et~al.}(2021)\citenamefont {Cirac},
  \citenamefont {P{\'{e}}rez-Garc{\'{\i}}a}, \citenamefont {Schuch},\ and\
  \citenamefont {Verstraete}}]{Cirac2021}%
  \BibitemOpen
  \bibfield  {author} {\bibinfo {author} {\bibfnamefont {J.~I.}\ \bibnamefont
  {Cirac}}, \bibinfo {author} {\bibfnamefont {D.}~\bibnamefont
  {P{\'{e}}rez-Garc{\'{\i}}a}}, \bibinfo {author} {\bibfnamefont
  {N.}~\bibnamefont {Schuch}},\ and\ \bibinfo {author} {\bibfnamefont
  {F.}~\bibnamefont {Verstraete}},\ }\bibfield  {title} {\bibinfo {title}
  {Matrix product states and projected entangled pair states: {C}oncepts,
  symmetries, theorems},\ }\href {https://doi.org/10.1103/revmodphys.93.045003}
  {\bibfield  {journal} {\bibinfo  {journal} {Rev. Mod. Phys.}\ }\textbf
  {\bibinfo {volume} {93}},\ \bibinfo {pages} {045003} (\bibinfo {year}
  {2021})},\ \Eprint {https://arxiv.org/abs/2011.12127} {arXiv:2011.12127
  [quant-ph]} \BibitemShut {NoStop}%
\bibitem [{\citenamefont {Poulin}\ \emph {et~al.}(2011)\citenamefont {Poulin},
  \citenamefont {Qarry}, \citenamefont {Somma},\ and\ \citenamefont
  {Verstraete}}]{Poulin2011}%
  \BibitemOpen
  \bibfield  {author} {\bibinfo {author} {\bibfnamefont {D.}~\bibnamefont
  {Poulin}}, \bibinfo {author} {\bibfnamefont {A.}~\bibnamefont {Qarry}},
  \bibinfo {author} {\bibfnamefont {R.}~\bibnamefont {Somma}},\ and\ \bibinfo
  {author} {\bibfnamefont {F.}~\bibnamefont {Verstraete}},\ }\bibfield  {title}
  {\bibinfo {title} {Quantum {S}imulation of {T}ime-{D}ependent {H}amiltonians
  and the {C}onvenient {I}llusion of {H}ilbert {S}pace},\ }\href
  {https://doi.org/10.1103/physrevlett.106.170501} {\bibfield  {journal}
  {\bibinfo  {journal} {Phys. Rev. Lett.}\ }\textbf {\bibinfo {volume} {106}},\
  \bibinfo {pages} {170501} (\bibinfo {year} {2011})},\ \Eprint
  {https://arxiv.org/abs/1102.1360} {arXiv:1102.1360 [quant-ph]} \BibitemShut
  {NoStop}%
\bibitem [{\citenamefont {Verstraete}\ and\ \citenamefont
  {Cirac}(2006)}]{Verstraete2006}%
  \BibitemOpen
  \bibfield  {author} {\bibinfo {author} {\bibfnamefont {F.}~\bibnamefont
  {Verstraete}}\ and\ \bibinfo {author} {\bibfnamefont {J.~I.}\ \bibnamefont
  {Cirac}},\ }\bibfield  {title} {\bibinfo {title} {Matrix product states
  represent ground states faithfully},\ }\href
  {https://doi.org/10.1103/physrevb.73.094423} {\bibfield  {journal} {\bibinfo
  {journal} {Phys. Rev. B}\ }\textbf {\bibinfo {volume} {73}},\ \bibinfo
  {pages} {094423} (\bibinfo {year} {2006})},\ \Eprint
  {https://arxiv.org/abs/cond-mat/0505140} {arXiv:cond-mat/0505140
  [cond-mat.str-el]} \BibitemShut {NoStop}%
\bibitem [{\citenamefont {Hastings}(2007)}]{Hastings2007}%
  \BibitemOpen
  \bibfield  {author} {\bibinfo {author} {\bibfnamefont {M.~B.}\ \bibnamefont
  {Hastings}},\ }\bibfield  {title} {\bibinfo {title} {An area law for
  one-dimensional quantum systems},\ }\href
  {https://doi.org/10.1088/1742-5468/2007/08/p08024} {\bibfield  {journal}
  {\bibinfo  {journal} {J. Stat. Mech.}\ }\textbf {\bibinfo {volume} {2007}},\
  \bibinfo {pages} {P08024} (\bibinfo {year} {2007})},\ \Eprint
  {https://arxiv.org/abs/0705.2024} {arXiv:0705.2024 [quant-ph]} \BibitemShut
  {NoStop}%
\bibitem [{\citenamefont {Arad}\ \emph {et~al.}()\citenamefont {Arad},
  \citenamefont {Kitaev}, \citenamefont {Landau},\ and\ \citenamefont
  {Vazirani}}]{Arad2013}%
  \BibitemOpen
  \bibfield  {author} {\bibinfo {author} {\bibfnamefont {I.}~\bibnamefont
  {Arad}}, \bibinfo {author} {\bibfnamefont {A.}~\bibnamefont {Kitaev}},
  \bibinfo {author} {\bibfnamefont {Z.}~\bibnamefont {Landau}},\ and\ \bibinfo
  {author} {\bibfnamefont {U.}~\bibnamefont {Vazirani}},\ }\href@noop {}
  {\bibinfo {title} {An area law and sub-exponential algorithm for 1{D}
  systems}},\ \Eprint {https://arxiv.org/abs/1301.1162} {arXiv:1301.1162
  [quant-ph]} \BibitemShut {NoStop}%
\bibitem [{\citenamefont {P{\'{e}}rez-Garc{\'{\i}}a}\ \emph
  {et~al.}(2008)\citenamefont {P{\'{e}}rez-Garc{\'{\i}}a}, \citenamefont
  {Verstraete}, \citenamefont {Cirac},\ and\ \citenamefont
  {Wolf}}]{PerezGarcia2008}%
  \BibitemOpen
  \bibfield  {author} {\bibinfo {author} {\bibfnamefont {D.}~\bibnamefont
  {P{\'{e}}rez-Garc{\'{\i}}a}}, \bibinfo {author} {\bibfnamefont
  {F.}~\bibnamefont {Verstraete}}, \bibinfo {author} {\bibfnamefont {J.~I.}\
  \bibnamefont {Cirac}},\ and\ \bibinfo {author} {\bibfnamefont {M.~M.}\
  \bibnamefont {Wolf}},\ }\bibfield  {title} {\bibinfo {title} {{PEPS} as
  unique ground states of local {H}amiltonians},\ }\href
  {https://doi.org/10.26421/qic8.6-7-6} {\bibfield  {journal} {\bibinfo
  {journal} {Quant. Inf. Comp.}\ }\textbf {\bibinfo {volume} {8}},\ \bibinfo
  {pages} {650} (\bibinfo {year} {2008})},\ \Eprint
  {https://arxiv.org/abs/0707.2260} {arXiv:0707.2260 [quant-ph]} \BibitemShut
  {NoStop}%
\bibitem [{\citenamefont {Daley}\ \emph {et~al.}(2004)\citenamefont {Daley},
  \citenamefont {Kollath}, \citenamefont {Schollw{\"{o}}ck},\ and\
  \citenamefont {Vidal}}]{Daley2004}%
  \BibitemOpen
  \bibfield  {author} {\bibinfo {author} {\bibfnamefont {A.~J.}\ \bibnamefont
  {Daley}}, \bibinfo {author} {\bibfnamefont {C.}~\bibnamefont {Kollath}},
  \bibinfo {author} {\bibfnamefont {U.}~\bibnamefont {Schollw{\"{o}}ck}},\ and\
  \bibinfo {author} {\bibfnamefont {G.}~\bibnamefont {Vidal}},\ }\bibfield
  {title} {\bibinfo {title} {Time-dependent density-matrix
  renormalization-group using adaptive effective {H}ilbert spaces},\ }\href
  {https://doi.org/10.1088/1742-5468/2004/04/p04005} {\bibfield  {journal}
  {\bibinfo  {journal} {J. Stat. Mech.}\ }\textbf {\bibinfo {volume} {2004}},\
  \bibinfo {pages} {P04005} (\bibinfo {year} {2004})},\ \Eprint
  {https://arxiv.org/abs/cond-mat/0403313} {arXiv:cond-mat/0403313
  [cond-mat.str-el]} \BibitemShut {NoStop}%
\bibitem [{\citenamefont {White}\ and\ \citenamefont
  {Feiguin}(2004)}]{White2004}%
  \BibitemOpen
  \bibfield  {author} {\bibinfo {author} {\bibfnamefont {S.~R.}\ \bibnamefont
  {White}}\ and\ \bibinfo {author} {\bibfnamefont {A.~E.}\ \bibnamefont
  {Feiguin}},\ }\bibfield  {title} {\bibinfo {title} {Real-{T}ime {E}volution
  {U}sing the {D}ensity {M}atrix {R}enormalization {G}roup},\ }\href
  {https://doi.org/10.1103/physrevlett.93.076401} {\bibfield  {journal}
  {\bibinfo  {journal} {Phys. Rev. Lett.}\ }\textbf {\bibinfo {volume} {93}},\
  \bibinfo {pages} {076401} (\bibinfo {year} {2004})},\ \Eprint
  {https://arxiv.org/abs/cond-mat/0403310} {arXiv:cond-mat/0403310
  [cond-mat.str-el]} \BibitemShut {NoStop}%
\bibitem [{\citenamefont {Vidal}(2007)}]{Vidal2007}%
  \BibitemOpen
  \bibfield  {author} {\bibinfo {author} {\bibfnamefont {G.}~\bibnamefont
  {Vidal}},\ }\bibfield  {title} {\bibinfo {title} {Classical {S}imulation of
  {I}nfinite-{S}ize {Q}uantum {L}attice {S}ystems in {O}ne {S}patial
  {D}imension},\ }\href {https://doi.org/10.1103/physrevlett.98.070201}
  {\bibfield  {journal} {\bibinfo  {journal} {Phys. Rev. Lett.}\ }\textbf
  {\bibinfo {volume} {98}},\ \bibinfo {pages} {070201} (\bibinfo {year}
  {2007})},\ \Eprint {https://arxiv.org/abs/cond-mat/0605597}
  {arXiv:cond-mat/0605597 [cond-mat.str-el]} \BibitemShut {NoStop}%
\bibitem [{\citenamefont {Or{\'{u}}s}\ and\ \citenamefont
  {Vidal}(2008)}]{Orus2008}%
  \BibitemOpen
  \bibfield  {author} {\bibinfo {author} {\bibfnamefont {R.}~\bibnamefont
  {Or{\'{u}}s}}\ and\ \bibinfo {author} {\bibfnamefont {G.}~\bibnamefont
  {Vidal}},\ }\bibfield  {title} {\bibinfo {title} {Infinite time-evolving
  block decimation algorithm beyond unitary evolution},\ }\href
  {https://doi.org/10.1103/physrevb.78.155117} {\bibfield  {journal} {\bibinfo
  {journal} {Phys. Rev. B}\ }\textbf {\bibinfo {volume} {78}},\ \bibinfo
  {pages} {155117} (\bibinfo {year} {2008})},\ \Eprint
  {https://arxiv.org/abs/0711.3960} {arXiv:0711.3960 [cond-mat.stat-mech]}
  \BibitemShut {NoStop}%
\bibitem [{\citenamefont {Ba{\~{n}}uls}\ \emph {et~al.}(2009)\citenamefont
  {Ba{\~{n}}uls}, \citenamefont {Hastings}, \citenamefont {Verstraete},\ and\
  \citenamefont {Cirac}}]{Banuls2009}%
  \BibitemOpen
  \bibfield  {author} {\bibinfo {author} {\bibfnamefont {M.~C.}\ \bibnamefont
  {Ba{\~{n}}uls}}, \bibinfo {author} {\bibfnamefont {M.~B.}\ \bibnamefont
  {Hastings}}, \bibinfo {author} {\bibfnamefont {F.}~\bibnamefont
  {Verstraete}},\ and\ \bibinfo {author} {\bibfnamefont {J.~I.}\ \bibnamefont
  {Cirac}},\ }\bibfield  {title} {\bibinfo {title} {Matrix {P}roduct {S}tates
  for {D}ynamical {S}imulation of {I}nfinite {C}hains},\ }\href
  {https://doi.org/10.1103/physrevlett.102.240603} {\bibfield  {journal}
  {\bibinfo  {journal} {Phys. Rev. Lett.}\ }\textbf {\bibinfo {volume} {102}},\
  \bibinfo {pages} {240603} (\bibinfo {year} {2009})},\ \Eprint
  {https://arxiv.org/abs/0904.1926} {arXiv:0904.1926 [quant-ph]} \BibitemShut
  {NoStop}%
\bibitem [{\citenamefont {Schollw{\"{o}}ck}(2011)}]{Schollwoeck2011}%
  \BibitemOpen
  \bibfield  {author} {\bibinfo {author} {\bibfnamefont {U.}~\bibnamefont
  {Schollw{\"{o}}ck}},\ }\bibfield  {title} {\bibinfo {title} {The
  density-matrix renormalization group in the age of matrix product states},\
  }\href {https://doi.org/10.1016/j.aop.2010.09.012} {\bibfield  {journal}
  {\bibinfo  {journal} {Ann. Phys.}\ }\textbf {\bibinfo {volume} {326}},\
  \bibinfo {pages} {96} (\bibinfo {year} {2011})},\ \Eprint
  {https://arxiv.org/abs/1008.3477} {arXiv:1008.3477 [cond-mat.str-el]}
  \BibitemShut {NoStop}%
\bibitem [{\citenamefont {Calabrese}\ and\ \citenamefont
  {Cardy}(2005)}]{Calabrese2005}%
  \BibitemOpen
  \bibfield  {author} {\bibinfo {author} {\bibfnamefont {P.}~\bibnamefont
  {Calabrese}}\ and\ \bibinfo {author} {\bibfnamefont {J.}~\bibnamefont
  {Cardy}},\ }\bibfield  {title} {\bibinfo {title} {Evolution of entanglement
  entropy in one-dimensional systems},\ }\href
  {https://doi.org/10.1088/1742-5468/2005/04/p04010} {\bibfield  {journal}
  {\bibinfo  {journal} {J. Stat. Mech.}\ }\textbf {\bibinfo {volume} {2005}},\
  \bibinfo {pages} {P04010} (\bibinfo {year} {2005})},\ \Eprint
  {https://arxiv.org/abs/cond-mat/0503393} {arXiv:cond-mat/0503393
  [cond-mat.stat-mech]} \BibitemShut {NoStop}%
\bibitem [{\citenamefont {Piroli}\ \emph {et~al.}(2020)\citenamefont {Piroli},
  \citenamefont {Bertini}, \citenamefont {Cirac},\ and\ \citenamefont
  {Prosen}}]{Piroli2020}%
  \BibitemOpen
  \bibfield  {author} {\bibinfo {author} {\bibfnamefont {L.}~\bibnamefont
  {Piroli}}, \bibinfo {author} {\bibfnamefont {B.}~\bibnamefont {Bertini}},
  \bibinfo {author} {\bibfnamefont {J.~I.}\ \bibnamefont {Cirac}},\ and\
  \bibinfo {author} {\bibfnamefont {T.}~\bibnamefont {Prosen}},\ }\bibfield
  {title} {\bibinfo {title} {Exact dynamics in dual-unitary quantum circuits},\
  }\href {https://doi.org/10.1103/physrevb.101.094304} {\bibfield  {journal}
  {\bibinfo  {journal} {Phys. Rev. B}\ }\textbf {\bibinfo {volume} {101}},\
  \bibinfo {pages} {094304} (\bibinfo {year} {2020})},\ \Eprint
  {https://arxiv.org/abs/1911.11175} {arXiv:1911.11175 [cond-mat.stat-mech]}
  \BibitemShut {NoStop}%
\bibitem [{\citenamefont {Bertini}\ \emph
  {et~al.}(2019{\natexlab{b}})\citenamefont {Bertini}, \citenamefont {Kos},\
  and\ \citenamefont {Prosen}}]{Bertini2019a}%
  \BibitemOpen
  \bibfield  {author} {\bibinfo {author} {\bibfnamefont {B.}~\bibnamefont
  {Bertini}}, \bibinfo {author} {\bibfnamefont {P.}~\bibnamefont {Kos}},\ and\
  \bibinfo {author} {\bibfnamefont {T.}~\bibnamefont {Prosen}},\ }\bibfield
  {title} {\bibinfo {title} {Entanglement {S}preading in a {M}inimal {M}odel of
  {M}aximal {M}any-{B}ody {Q}uantum {C}haos},\ }\href
  {https://doi.org/10.1103/physrevx.9.021033} {\bibfield  {journal} {\bibinfo
  {journal} {Phys. Rev. X}\ }\textbf {\bibinfo {volume} {9}},\ \bibinfo {pages}
  {021033} (\bibinfo {year} {2019}{\natexlab{b}})},\ \Eprint
  {https://arxiv.org/abs/1812.05090} {arXiv:1812.05090 [cond-mat.stat-mech]}
  \BibitemShut {NoStop}%
\bibitem [{\citenamefont {Gopalakrishnan}\ and\ \citenamefont
  {Lamacraft}(2019)}]{Gopalakrishnan2019}%
  \BibitemOpen
  \bibfield  {author} {\bibinfo {author} {\bibfnamefont {S.}~\bibnamefont
  {Gopalakrishnan}}\ and\ \bibinfo {author} {\bibfnamefont {A.}~\bibnamefont
  {Lamacraft}},\ }\bibfield  {title} {\bibinfo {title} {Unitary circuits of
  finite depth and infinite width from quantum channels},\ }\href
  {https://doi.org/10.1103/physrevb.100.064309} {\bibfield  {journal} {\bibinfo
   {journal} {Phys. Rev. B}\ }\textbf {\bibinfo {volume} {100}},\ \bibinfo
  {pages} {064309} (\bibinfo {year} {2019})},\ \Eprint
  {https://arxiv.org/abs/1903.11611} {arXiv:1903.11611 [quant-ph]} \BibitemShut
  {NoStop}%
\bibitem [{\citenamefont {Claeys}\ and\ \citenamefont
  {Lamacraft}(2020)}]{Claeys2020}%
  \BibitemOpen
  \bibfield  {author} {\bibinfo {author} {\bibfnamefont {P.~W.}\ \bibnamefont
  {Claeys}}\ and\ \bibinfo {author} {\bibfnamefont {A.}~\bibnamefont
  {Lamacraft}},\ }\bibfield  {title} {\bibinfo {title} {Maximum velocity
  quantum circuits},\ }\href {https://doi.org/10.1103/physrevresearch.2.033032}
  {\bibfield  {journal} {\bibinfo  {journal} {Phys. Rev. Research}\ }\textbf
  {\bibinfo {volume} {2}},\ \bibinfo {pages} {033032} (\bibinfo {year}
  {2020})},\ \Eprint {https://arxiv.org/abs/2003.01133} {arXiv:2003.01133
  [quant-ph]} \BibitemShut {NoStop}%
\bibitem [{\citenamefont {Bertini}\ \emph {et~al.}(2020)\citenamefont
  {Bertini}, \citenamefont {Kos},\ and\ \citenamefont {Prosen}}]{Bertini2020}%
  \BibitemOpen
  \bibfield  {author} {\bibinfo {author} {\bibfnamefont {B.}~\bibnamefont
  {Bertini}}, \bibinfo {author} {\bibfnamefont {P.}~\bibnamefont {Kos}},\ and\
  \bibinfo {author} {\bibfnamefont {T.}~\bibnamefont {Prosen}},\ }\bibfield
  {title} {\bibinfo {title} {Operator {E}ntanglement in {L}ocal {Q}uantum
  {C}ircuits {I}: {C}haotic {D}ual-{U}nitary {C}ircuits},\ }\href
  {https://doi.org/10.21468/scipostphys.8.4.067} {\bibfield  {journal}
  {\bibinfo  {journal} {SciPost Phys.}\ }\textbf {\bibinfo {volume} {8}},\
  \bibinfo {pages} {067} (\bibinfo {year} {2020})},\ \Eprint
  {https://arxiv.org/abs/1909.07407} {arXiv:1909.07407 [cond-mat.stat-mech]}
  \BibitemShut {NoStop}%
\bibitem [{\citenamefont {Bertini}\ and\ \citenamefont
  {Piroli}(2020)}]{Bertini2020a}%
  \BibitemOpen
  \bibfield  {author} {\bibinfo {author} {\bibfnamefont {B.}~\bibnamefont
  {Bertini}}\ and\ \bibinfo {author} {\bibfnamefont {L.}~\bibnamefont
  {Piroli}},\ }\bibfield  {title} {\bibinfo {title} {Scrambling in random
  unitary circuits: {E}xact results},\ }\href
  {https://doi.org/10.1103/physrevb.102.064305} {\bibfield  {journal} {\bibinfo
   {journal} {Phys. Rev. B}\ }\textbf {\bibinfo {volume} {102}},\ \bibinfo
  {pages} {064305} (\bibinfo {year} {2020})},\ \Eprint
  {https://arxiv.org/abs/2004.13697} {arXiv:2004.13697 [cond-mat.stat-mech]}
  \BibitemShut {NoStop}%
\bibitem [{\citenamefont {Mi}\ \emph {et~al.}(2021)\citenamefont {Mi},
  \citenamefont {Roushan}, \citenamefont {Quintana}, \citenamefont
  {Mandr{\`{a}}}, \citenamefont {Marshall}, \citenamefont {Neill},
  \citenamefont {Arute}, \citenamefont {Arya}, \citenamefont {Atalaya},
  \citenamefont {Babbush}, \citenamefont {Bardin}, \citenamefont {Barends},
  \citenamefont {Basso}, \citenamefont {Bengtsson}, \citenamefont {Boixo},
  \citenamefont {Bourassa}, \citenamefont {Broughton}, \citenamefont {Buckley},
  \citenamefont {Buell}, \citenamefont {Burkett}, \citenamefont {Bushnell},
  \citenamefont {Chen}, \citenamefont {Chiaro}, \citenamefont {Collins},
  \citenamefont {Courtney}, \citenamefont {Demura}, \citenamefont {Derk},
  \citenamefont {Dunsworth}, \citenamefont {Eppens}, \citenamefont {Erickson},
  \citenamefont {Farhi}, \citenamefont {Fowler}, \citenamefont {Foxen},
  \citenamefont {Gidney}, \citenamefont {Giustina}, \citenamefont {Gross},
  \citenamefont {Harrigan}, \citenamefont {Harrington}, \citenamefont {Hilton},
  \citenamefont {Ho}, \citenamefont {Hong}, \citenamefont {Huang},
  \citenamefont {Huggins}, \citenamefont {Ioffe}, \citenamefont {Isakov},
  \citenamefont {Jeffrey}, \citenamefont {Jiang}, \citenamefont {Jones},
  \citenamefont {Kafri}, \citenamefont {Kelly}, \citenamefont {Kim},
  \citenamefont {Kitaev}, \citenamefont {Klimov}, \citenamefont {Korotkov},
  \citenamefont {Kostritsa}, \citenamefont {Landhuis}, \citenamefont {Laptev},
  \citenamefont {Lucero}, \citenamefont {Martin}, \citenamefont {McClean},
  \citenamefont {McCourt}, \citenamefont {McEwen}, \citenamefont {Megrant},
  \citenamefont {Miao}, \citenamefont {Mohseni}, \citenamefont {Montazeri},
  \citenamefont {Mruczkiewicz}, \citenamefont {Mutus}, \citenamefont {Naaman},
  \citenamefont {Neeley}, \citenamefont {Newman}, \citenamefont {Niu},
  \citenamefont {O'Brien}, \citenamefont {Opremcak}, \citenamefont {Ostby},
  \citenamefont {Pato}, \citenamefont {Petukhov}, \citenamefont {Redd},
  \citenamefont {Rubin}, \citenamefont {Sank}, \citenamefont {Satzinger},
  \citenamefont {Shvarts}, \citenamefont {Strain}, \citenamefont {Szalay},
  \citenamefont {Trevithick}, \citenamefont {Villalonga}, \citenamefont
  {White}, \citenamefont {Yao}, \citenamefont {Yeh}, \citenamefont {Zalcman},
  \citenamefont {Neven}, \citenamefont {Aleiner}, \citenamefont {Kechedzhi},
  \citenamefont {Smelyanskiy},\ and\ \citenamefont {Chen}}]{Mi2021}%
  \BibitemOpen
  \bibfield  {author} {\bibinfo {author} {\bibfnamefont {X.}~\bibnamefont
  {Mi}}, \bibinfo {author} {\bibfnamefont {P.}~\bibnamefont {Roushan}},
  \bibinfo {author} {\bibfnamefont {C.}~\bibnamefont {Quintana}}, \bibinfo
  {author} {\bibfnamefont {S.}~\bibnamefont {Mandr{\`{a}}}}, \bibinfo {author}
  {\bibfnamefont {J.}~\bibnamefont {Marshall}}, \bibinfo {author}
  {\bibfnamefont {C.}~\bibnamefont {Neill}}, \bibinfo {author} {\bibfnamefont
  {F.}~\bibnamefont {Arute}}, \bibinfo {author} {\bibfnamefont
  {K.}~\bibnamefont {Arya}}, \bibinfo {author} {\bibfnamefont {J.}~\bibnamefont
  {Atalaya}}, \bibinfo {author} {\bibfnamefont {R.}~\bibnamefont {Babbush}},
  \bibinfo {author} {\bibfnamefont {J.~C.}\ \bibnamefont {Bardin}}, \bibinfo
  {author} {\bibfnamefont {R.}~\bibnamefont {Barends}}, \bibinfo {author}
  {\bibfnamefont {J.}~\bibnamefont {Basso}}, \bibinfo {author} {\bibfnamefont
  {A.}~\bibnamefont {Bengtsson}}, \bibinfo {author} {\bibfnamefont
  {S.}~\bibnamefont {Boixo}}, \bibinfo {author} {\bibfnamefont
  {A.}~\bibnamefont {Bourassa}}, \bibinfo {author} {\bibfnamefont
  {M.}~\bibnamefont {Broughton}}, \bibinfo {author} {\bibfnamefont {B.~B.}\
  \bibnamefont {Buckley}}, \bibinfo {author} {\bibfnamefont {D.~A.}\
  \bibnamefont {Buell}}, \bibinfo {author} {\bibfnamefont {B.}~\bibnamefont
  {Burkett}}, \bibinfo {author} {\bibfnamefont {N.}~\bibnamefont {Bushnell}},
  \bibinfo {author} {\bibfnamefont {Z.}~\bibnamefont {Chen}}, \bibinfo {author}
  {\bibfnamefont {B.}~\bibnamefont {Chiaro}}, \bibinfo {author} {\bibfnamefont
  {R.}~\bibnamefont {Collins}}, \bibinfo {author} {\bibfnamefont
  {W.}~\bibnamefont {Courtney}}, \bibinfo {author} {\bibfnamefont
  {S.}~\bibnamefont {Demura}}, \bibinfo {author} {\bibfnamefont {A.~R.}\
  \bibnamefont {Derk}}, \bibinfo {author} {\bibfnamefont {A.}~\bibnamefont
  {Dunsworth}}, \bibinfo {author} {\bibfnamefont {D.}~\bibnamefont {Eppens}},
  \bibinfo {author} {\bibfnamefont {C.}~\bibnamefont {Erickson}}, \bibinfo
  {author} {\bibfnamefont {E.}~\bibnamefont {Farhi}}, \bibinfo {author}
  {\bibfnamefont {A.~G.}\ \bibnamefont {Fowler}}, \bibinfo {author}
  {\bibfnamefont {B.}~\bibnamefont {Foxen}}, \bibinfo {author} {\bibfnamefont
  {C.}~\bibnamefont {Gidney}}, \bibinfo {author} {\bibfnamefont
  {M.}~\bibnamefont {Giustina}}, \bibinfo {author} {\bibfnamefont {J.~A.}\
  \bibnamefont {Gross}}, \bibinfo {author} {\bibfnamefont {M.~P.}\ \bibnamefont
  {Harrigan}}, \bibinfo {author} {\bibfnamefont {S.~D.}\ \bibnamefont
  {Harrington}}, \bibinfo {author} {\bibfnamefont {J.}~\bibnamefont {Hilton}},
  \bibinfo {author} {\bibfnamefont {A.}~\bibnamefont {Ho}}, \bibinfo {author}
  {\bibfnamefont {S.}~\bibnamefont {Hong}}, \bibinfo {author} {\bibfnamefont
  {T.}~\bibnamefont {Huang}}, \bibinfo {author} {\bibfnamefont {W.~J.}\
  \bibnamefont {Huggins}}, \bibinfo {author} {\bibfnamefont {L.~B.}\
  \bibnamefont {Ioffe}}, \bibinfo {author} {\bibfnamefont {S.~V.}\ \bibnamefont
  {Isakov}}, \bibinfo {author} {\bibfnamefont {E.}~\bibnamefont {Jeffrey}},
  \bibinfo {author} {\bibfnamefont {Z.}~\bibnamefont {Jiang}}, \bibinfo
  {author} {\bibfnamefont {C.}~\bibnamefont {Jones}}, \bibinfo {author}
  {\bibfnamefont {D.}~\bibnamefont {Kafri}}, \bibinfo {author} {\bibfnamefont
  {J.}~\bibnamefont {Kelly}}, \bibinfo {author} {\bibfnamefont
  {S.}~\bibnamefont {Kim}}, \bibinfo {author} {\bibfnamefont {A.}~\bibnamefont
  {Kitaev}}, \bibinfo {author} {\bibfnamefont {P.~V.}\ \bibnamefont {Klimov}},
  \bibinfo {author} {\bibfnamefont {A.~N.}\ \bibnamefont {Korotkov}}, \bibinfo
  {author} {\bibfnamefont {F.}~\bibnamefont {Kostritsa}}, \bibinfo {author}
  {\bibfnamefont {D.}~\bibnamefont {Landhuis}}, \bibinfo {author}
  {\bibfnamefont {P.}~\bibnamefont {Laptev}}, \bibinfo {author} {\bibfnamefont
  {E.}~\bibnamefont {Lucero}}, \bibinfo {author} {\bibfnamefont
  {O.}~\bibnamefont {Martin}}, \bibinfo {author} {\bibfnamefont {J.~R.}\
  \bibnamefont {McClean}}, \bibinfo {author} {\bibfnamefont {T.}~\bibnamefont
  {McCourt}}, \bibinfo {author} {\bibfnamefont {M.}~\bibnamefont {McEwen}},
  \bibinfo {author} {\bibfnamefont {A.}~\bibnamefont {Megrant}}, \bibinfo
  {author} {\bibfnamefont {K.~C.}\ \bibnamefont {Miao}}, \bibinfo {author}
  {\bibfnamefont {M.}~\bibnamefont {Mohseni}}, \bibinfo {author} {\bibfnamefont
  {S.}~\bibnamefont {Montazeri}}, \bibinfo {author} {\bibfnamefont
  {W.}~\bibnamefont {Mruczkiewicz}}, \bibinfo {author} {\bibfnamefont
  {J.}~\bibnamefont {Mutus}}, \bibinfo {author} {\bibfnamefont
  {O.}~\bibnamefont {Naaman}}, \bibinfo {author} {\bibfnamefont
  {M.}~\bibnamefont {Neeley}}, \bibinfo {author} {\bibfnamefont
  {M.}~\bibnamefont {Newman}}, \bibinfo {author} {\bibfnamefont {M.~Y.}\
  \bibnamefont {Niu}}, \bibinfo {author} {\bibfnamefont {T.~E.}\ \bibnamefont
  {O'Brien}}, \bibinfo {author} {\bibfnamefont {A.}~\bibnamefont {Opremcak}},
  \bibinfo {author} {\bibfnamefont {E.}~\bibnamefont {Ostby}}, \bibinfo
  {author} {\bibfnamefont {B.}~\bibnamefont {Pato}}, \bibinfo {author}
  {\bibfnamefont {A.}~\bibnamefont {Petukhov}}, \bibinfo {author}
  {\bibfnamefont {N.}~\bibnamefont {Redd}}, \bibinfo {author} {\bibfnamefont
  {N.~C.}\ \bibnamefont {Rubin}}, \bibinfo {author} {\bibfnamefont
  {D.}~\bibnamefont {Sank}}, \bibinfo {author} {\bibfnamefont {K.~J.}\
  \bibnamefont {Satzinger}}, \bibinfo {author} {\bibfnamefont {V.}~\bibnamefont
  {Shvarts}}, \bibinfo {author} {\bibfnamefont {D.}~\bibnamefont {Strain}},
  \bibinfo {author} {\bibfnamefont {M.}~\bibnamefont {Szalay}}, \bibinfo
  {author} {\bibfnamefont {M.~D.}\ \bibnamefont {Trevithick}}, \bibinfo
  {author} {\bibfnamefont {B.}~\bibnamefont {Villalonga}}, \bibinfo {author}
  {\bibfnamefont {T.}~\bibnamefont {White}}, \bibinfo {author} {\bibfnamefont
  {Z.~J.}\ \bibnamefont {Yao}}, \bibinfo {author} {\bibfnamefont
  {P.}~\bibnamefont {Yeh}}, \bibinfo {author} {\bibfnamefont {A.}~\bibnamefont
  {Zalcman}}, \bibinfo {author} {\bibfnamefont {H.}~\bibnamefont {Neven}},
  \bibinfo {author} {\bibfnamefont {I.}~\bibnamefont {Aleiner}}, \bibinfo
  {author} {\bibfnamefont {K.}~\bibnamefont {Kechedzhi}}, \bibinfo {author}
  {\bibfnamefont {V.}~\bibnamefont {Smelyanskiy}},\ and\ \bibinfo {author}
  {\bibfnamefont {Y.}~\bibnamefont {Chen}},\ }\bibfield  {title} {\bibinfo
  {title} {Information scrambling in quantum circuits},\ }\href
  {https://doi.org/10.1126/science.abg5029} {\bibfield  {journal} {\bibinfo
  {journal} {Science}\ }\textbf {\bibinfo {volume} {374}},\ \bibinfo {pages}
  {1479} (\bibinfo {year} {2021})},\ \Eprint {https://arxiv.org/abs/2101.08870}
  {arXiv:2101.08870 [quant-ph]} \BibitemShut {NoStop}%
\bibitem [{\citenamefont {Chertkov}\ \emph {et~al.}(2022)\citenamefont
  {Chertkov}, \citenamefont {Bohnet}, \citenamefont {Francois}, \citenamefont
  {Gaebler}, \citenamefont {Gresh}, \citenamefont {Hankin}, \citenamefont
  {Lee}, \citenamefont {Hayes}, \citenamefont {Neyenhuis}, \citenamefont
  {Stutz}, \citenamefont {Potter},\ and\ \citenamefont
  {Foss-Feig}}]{Chertkov2022}%
  \BibitemOpen
  \bibfield  {author} {\bibinfo {author} {\bibfnamefont {E.}~\bibnamefont
  {Chertkov}}, \bibinfo {author} {\bibfnamefont {J.}~\bibnamefont {Bohnet}},
  \bibinfo {author} {\bibfnamefont {D.}~\bibnamefont {Francois}}, \bibinfo
  {author} {\bibfnamefont {J.}~\bibnamefont {Gaebler}}, \bibinfo {author}
  {\bibfnamefont {D.}~\bibnamefont {Gresh}}, \bibinfo {author} {\bibfnamefont
  {A.}~\bibnamefont {Hankin}}, \bibinfo {author} {\bibfnamefont
  {K.}~\bibnamefont {Lee}}, \bibinfo {author} {\bibfnamefont {D.}~\bibnamefont
  {Hayes}}, \bibinfo {author} {\bibfnamefont {B.}~\bibnamefont {Neyenhuis}},
  \bibinfo {author} {\bibfnamefont {R.}~\bibnamefont {Stutz}}, \bibinfo
  {author} {\bibfnamefont {A.~C.}\ \bibnamefont {Potter}},\ and\ \bibinfo
  {author} {\bibfnamefont {M.}~\bibnamefont {Foss-Feig}},\ }\bibfield  {title}
  {\bibinfo {title} {Holographic dynamics simulations with a trapped-ion
  quantum computer},\ }\href {https://doi.org/10.1038/s41567-022-01689-7}
  {\bibfield  {journal} {\bibinfo  {journal} {Nat. Phys.}\ }\textbf {\bibinfo
  {volume} {18}},\ \bibinfo {pages} {1074} (\bibinfo {year} {2022})},\ \Eprint
  {https://arxiv.org/abs/2105.09324} {arXiv:2105.09324 [quant-ph]} \BibitemShut
  {NoStop}%
\bibitem [{\citenamefont {Prosen}(2021)}]{Prosen2021}%
  \BibitemOpen
  \bibfield  {author} {\bibinfo {author} {\bibfnamefont {T.}~\bibnamefont
  {Prosen}},\ }\bibfield  {title} {\bibinfo {title} {Many-body quantum chaos
  and dual-unitarity round-a-face},\ }\href {https://doi.org/10.1063/5.0056970}
  {\bibfield  {journal} {\bibinfo  {journal} {Chaos}\ }\textbf {\bibinfo
  {volume} {31}},\ \bibinfo {pages} {093101} (\bibinfo {year} {2021})},\
  \Eprint {https://arxiv.org/abs/2105.08022} {arXiv:2105.08022
  [cond-mat.stat-mech]} \BibitemShut {NoStop}%
\bibitem [{\citenamefont {Jonay}\ \emph {et~al.}(2021)\citenamefont {Jonay},
  \citenamefont {Khemani},\ and\ \citenamefont {Ippoliti}}]{Jonay2021}%
  \BibitemOpen
  \bibfield  {author} {\bibinfo {author} {\bibfnamefont {C.}~\bibnamefont
  {Jonay}}, \bibinfo {author} {\bibfnamefont {V.}~\bibnamefont {Khemani}},\
  and\ \bibinfo {author} {\bibfnamefont {M.}~\bibnamefont {Ippoliti}},\
  }\bibfield  {title} {\bibinfo {title} {Triunitary quantum circuits},\ }\href
  {https://doi.org/10.1103/physrevresearch.3.043046} {\bibfield  {journal}
  {\bibinfo  {journal} {Phys. Rev. Research}\ }\textbf {\bibinfo {volume}
  {3}},\ \bibinfo {pages} {043046} (\bibinfo {year} {2021})},\ \Eprint
  {https://arxiv.org/abs/2106.07686} {arXiv:2106.07686 [quant-ph]} \BibitemShut
  {NoStop}%
\bibitem [{\citenamefont {Mesty{\'{a}}n}\ \emph {et~al.}()\citenamefont
  {Mesty{\'{a}}n}, \citenamefont {Pozsgay},\ and\ \citenamefont
  {Wanless}}]{Mestyan2022}%
  \BibitemOpen
  \bibfield  {author} {\bibinfo {author} {\bibfnamefont {M.}~\bibnamefont
  {Mesty{\'{a}}n}}, \bibinfo {author} {\bibfnamefont {B.}~\bibnamefont
  {Pozsgay}},\ and\ \bibinfo {author} {\bibfnamefont {I.~M.}\ \bibnamefont
  {Wanless}},\ }\href@noop {} {\bibinfo {title} {Multi-directional unitarity
  and maximal entanglement in spatially symmetric quantum states}},\ \Eprint
  {https://arxiv.org/abs/2210.13017} {arXiv:2210.13017 [quant-ph]} \BibitemShut
  {NoStop}%
\bibitem [{\citenamefont {Kos}\ and\ \citenamefont
  {Styliaris}(2023)}]{Kos2023}%
  \BibitemOpen
  \bibfield  {author} {\bibinfo {author} {\bibfnamefont {P.}~\bibnamefont
  {Kos}}\ and\ \bibinfo {author} {\bibfnamefont {G.}~\bibnamefont
  {Styliaris}},\ }\bibfield  {title} {\bibinfo {title} {Circuits of space and
  time quantum channels},\ }\href {https://doi.org/10.22331/q-2023-05-24-1020}
  {\bibfield  {journal} {\bibinfo  {journal} {Quantum}\ }\textbf {\bibinfo
  {volume} {7}},\ \bibinfo {pages} {1020} (\bibinfo {year} {2023})},\ \Eprint
  {https://arxiv.org/abs/2206.12155} {arXiv:2206.12155 [cond-mat.stat-mech]}
  \BibitemShut {NoStop}%
\bibitem [{\citenamefont {Milbradt}\ \emph {et~al.}(2023)\citenamefont
  {Milbradt}, \citenamefont {Scheller}, \citenamefont {A{\ss}mus},\ and\
  \citenamefont {Mendl}}]{Milbradt2023}%
  \BibitemOpen
  \bibfield  {author} {\bibinfo {author} {\bibfnamefont {R.~M.}\ \bibnamefont
  {Milbradt}}, \bibinfo {author} {\bibfnamefont {L.}~\bibnamefont {Scheller}},
  \bibinfo {author} {\bibfnamefont {C.}~\bibnamefont {A{\ss}mus}},\ and\
  \bibinfo {author} {\bibfnamefont {C.~B.}\ \bibnamefont {Mendl}},\ }\bibfield
  {title} {\bibinfo {title} {Ternary {U}nitary {Q}uantum {L}attice {M}odels and
  {C}ircuits in $2 + 1$ {D}imensions},\ }\href
  {https://doi.org/10.1103/physrevlett.130.090601} {\bibfield  {journal}
  {\bibinfo  {journal} {Phys. Rev. Lett.}\ }\textbf {\bibinfo {volume} {130}},\
  \bibinfo {pages} {090601} (\bibinfo {year} {2023})},\ \Eprint
  {https://arxiv.org/abs/2206.01499} {arXiv:2206.01499 [cond-mat.stat-mech]}
  \BibitemShut {NoStop}%
\bibitem [{\citenamefont {Kasim}\ and\ \citenamefont
  {Prosen}(2023)}]{Kasim2023}%
  \BibitemOpen
  \bibfield  {author} {\bibinfo {author} {\bibfnamefont {Y.}~\bibnamefont
  {Kasim}}\ and\ \bibinfo {author} {\bibfnamefont {T.}~\bibnamefont {Prosen}},\
  }\bibfield  {title} {\bibinfo {title} {Dual unitary circuits in random
  geometries},\ }\href {https://doi.org/10.1088/1751-8121/acb1e0} {\bibfield
  {journal} {\bibinfo  {journal} {J. Phys. A: Math. Theor.}\ }\textbf {\bibinfo
  {volume} {56}},\ \bibinfo {pages} {025003} (\bibinfo {year} {2023})},\
  \Eprint {https://arxiv.org/abs/2206.09665} {arXiv:2206.09665
  [cond-mat.stat-mech]} \BibitemShut {NoStop}%
\bibitem [{\citenamefont {Garnerone}\ \emph
  {et~al.}(2010{\natexlab{a}})\citenamefont {Garnerone}, \citenamefont
  {de~Oliveira},\ and\ \citenamefont {Zanardi}}]{Garnerone2010}%
  \BibitemOpen
  \bibfield  {author} {\bibinfo {author} {\bibfnamefont {S.}~\bibnamefont
  {Garnerone}}, \bibinfo {author} {\bibfnamefont {T.~R.}\ \bibnamefont
  {de~Oliveira}},\ and\ \bibinfo {author} {\bibfnamefont {P.}~\bibnamefont
  {Zanardi}},\ }\bibfield  {title} {\bibinfo {title} {Typicality in random
  matrix product states},\ }\href {https://doi.org/10.1103/physreva.81.032336}
  {\bibfield  {journal} {\bibinfo  {journal} {Phys. Rev. A}\ }\textbf {\bibinfo
  {volume} {81}},\ \bibinfo {pages} {032336} (\bibinfo {year}
  {2010}{\natexlab{a}})},\ \Eprint {https://arxiv.org/abs/0908.3877}
  {arXiv:0908.3877 [quant-ph]} \BibitemShut {NoStop}%
\bibitem [{\citenamefont {Garnerone}\ \emph
  {et~al.}(2010{\natexlab{b}})\citenamefont {Garnerone}, \citenamefont
  {de~Oliveira}, \citenamefont {Haas},\ and\ \citenamefont
  {Zanardi}}]{Garnerone2010a}%
  \BibitemOpen
  \bibfield  {author} {\bibinfo {author} {\bibfnamefont {S.}~\bibnamefont
  {Garnerone}}, \bibinfo {author} {\bibfnamefont {T.~R.}\ \bibnamefont
  {de~Oliveira}}, \bibinfo {author} {\bibfnamefont {S.}~\bibnamefont {Haas}},\
  and\ \bibinfo {author} {\bibfnamefont {P.}~\bibnamefont {Zanardi}},\
  }\bibfield  {title} {\bibinfo {title} {Statistical properties of random
  matrix product states},\ }\href {https://doi.org/10.1103/physreva.82.052312}
  {\bibfield  {journal} {\bibinfo  {journal} {Phys. Rev. A}\ }\textbf {\bibinfo
  {volume} {82}},\ \bibinfo {pages} {052312} (\bibinfo {year}
  {2010}{\natexlab{b}})},\ \Eprint {https://arxiv.org/abs/1003.5253}
  {arXiv:1003.5253 [quant-ph]} \BibitemShut {NoStop}%
\bibitem [{\citenamefont {Haferkamp}\ \emph {et~al.}(2021)\citenamefont
  {Haferkamp}, \citenamefont {Bertoni}, \citenamefont {Roth},\ and\
  \citenamefont {Eisert}}]{Haferkamp2021}%
  \BibitemOpen
  \bibfield  {author} {\bibinfo {author} {\bibfnamefont {J.}~\bibnamefont
  {Haferkamp}}, \bibinfo {author} {\bibfnamefont {C.}~\bibnamefont {Bertoni}},
  \bibinfo {author} {\bibfnamefont {I.}~\bibnamefont {Roth}},\ and\ \bibinfo
  {author} {\bibfnamefont {J.}~\bibnamefont {Eisert}},\ }\bibfield  {title}
  {\bibinfo {title} {Emergent {S}tatistical {M}echanics from {P}roperties of
  {D}isordered {R}andom {M}atrix {P}roduct {S}tates},\ }\href
  {https://doi.org/10.1103/prxquantum.2.040308} {\bibfield  {journal} {\bibinfo
   {journal} {PRX Quantum}\ }\textbf {\bibinfo {volume} {2}},\ \bibinfo {pages}
  {040308} (\bibinfo {year} {2021})},\ \Eprint
  {https://arxiv.org/abs/2103.02634} {arXiv:2103.02634 [quant-ph]} \BibitemShut
  {NoStop}%
\bibitem [{\citenamefont {Collins}\ \emph {et~al.}(2013)\citenamefont
  {Collins}, \citenamefont {Gonz{\'{a}}lez-Guill{\'{e}}n},\ and\ \citenamefont
  {P{\'{e}}rez-Garc{\'{\i}}a}}]{Collins2013}%
  \BibitemOpen
  \bibfield  {author} {\bibinfo {author} {\bibfnamefont {B.}~\bibnamefont
  {Collins}}, \bibinfo {author} {\bibfnamefont {C.~E.}\ \bibnamefont
  {Gonz{\'{a}}lez-Guill{\'{e}}n}},\ and\ \bibinfo {author} {\bibfnamefont
  {D.}~\bibnamefont {P{\'{e}}rez-Garc{\'{\i}}a}},\ }\bibfield  {title}
  {\bibinfo {title} {Matrix {P}roduct {S}tates, {R}andom {M}atrix {T}heory and
  the {P}rinciple of {M}aximum {E}ntropy},\ }\href
  {https://doi.org/10.1007/s00220-013-1718-x} {\bibfield  {journal} {\bibinfo
  {journal} {Comm. Math. Phys.}\ }\textbf {\bibinfo {volume} {320}},\ \bibinfo
  {pages} {663} (\bibinfo {year} {2013})},\ \Eprint
  {https://arxiv.org/abs/1201.6324} {arXiv:1201.6324 [quant-ph]} \BibitemShut
  {NoStop}%
\bibitem [{\citenamefont {Chen}\ \emph {et~al.}()\citenamefont {Chen},
  \citenamefont {Garcia}, \citenamefont {Bu},\ and\ \citenamefont
  {Jaffe}}]{Chen2022}%
  \BibitemOpen
  \bibfield  {author} {\bibinfo {author} {\bibfnamefont {L.}~\bibnamefont
  {Chen}}, \bibinfo {author} {\bibfnamefont {R.~J.}\ \bibnamefont {Garcia}},
  \bibinfo {author} {\bibfnamefont {K.}~\bibnamefont {Bu}},\ and\ \bibinfo
  {author} {\bibfnamefont {A.}~\bibnamefont {Jaffe}},\ }\href@noop {} {\bibinfo
  {title} {Magic of {R}andom {M}atrix {P}roduct {S}tates}},\ \Eprint
  {https://arxiv.org/abs/2211.10350} {arXiv:2211.10350 [quant-ph]} \BibitemShut
  {NoStop}%
\bibitem [{\citenamefont {Movassagh}(2017)}]{Movassagh2017}%
  \BibitemOpen
  \bibfield  {author} {\bibinfo {author} {\bibfnamefont {R.}~\bibnamefont
  {Movassagh}},\ }\bibfield  {title} {\bibinfo {title} {Generic {L}ocal
  {H}amiltonians are {G}apless},\ }\href
  {https://doi.org/10.1103/physrevlett.119.220504} {\bibfield  {journal}
  {\bibinfo  {journal} {Phys. Rev. Lett.}\ }\textbf {\bibinfo {volume} {119}},\
  \bibinfo {pages} {220504} (\bibinfo {year} {2017})},\ \Eprint
  {https://arxiv.org/abs/1606.09313} {arXiv:1606.09313 [quant-ph]} \BibitemShut
  {NoStop}%
\bibitem [{\citenamefont {Movassagh}\ and\ \citenamefont
  {Schenker}(2022)}]{Movassagh2022}%
  \BibitemOpen
  \bibfield  {author} {\bibinfo {author} {\bibfnamefont {R.}~\bibnamefont
  {Movassagh}}\ and\ \bibinfo {author} {\bibfnamefont {J.}~\bibnamefont
  {Schenker}},\ }\bibfield  {title} {\bibinfo {title} {An {E}rgodic {T}heorem
  for {Q}uantum {P}rocesses with {A}pplications to {M}atrix {P}roduct
  {S}tates},\ }\href {https://doi.org/10.1007/s00220-022-04448-0} {\bibfield
  {journal} {\bibinfo  {journal} {Comm. Math. Phys.}\ }\textbf {\bibinfo
  {volume} {395}},\ \bibinfo {pages} {1175} (\bibinfo {year} {2022})},\ \Eprint
  {https://arxiv.org/abs/1909.11769} {arXiv:1909.11769 [quant-ph]} \BibitemShut
  {NoStop}%
\bibitem [{\citenamefont {Lancien}\ and\ \citenamefont
  {P{\'{e}}rez-Garc{\'{i}}a}(2021)}]{Lancien2021}%
  \BibitemOpen
  \bibfield  {author} {\bibinfo {author} {\bibfnamefont {C.}~\bibnamefont
  {Lancien}}\ and\ \bibinfo {author} {\bibfnamefont {D.}~\bibnamefont
  {P{\'{e}}rez-Garc{\'{i}}a}},\ }\bibfield  {title} {\bibinfo {title}
  {Correlation {L}ength in {R}andom {MPS} and {PEPS}},\ }\href
  {https://doi.org/10.1007/s00023-021-01087-4} {\bibfield  {journal} {\bibinfo
  {journal} {Ann. Henri Poincar{\'{e}}}\ }\textbf {\bibinfo {volume} {23}},\
  \bibinfo {pages} {141} (\bibinfo {year} {2021})},\ \Eprint
  {https://arxiv.org/abs/1906.11682} {arXiv:1906.11682 [quant-ph]} \BibitemShut
  {NoStop}%
\bibitem [{\citenamefont {Bensa}\ and\ \citenamefont
  {{\v{Z}}nidari{\v{c}}}(2023)}]{Bensa2023}%
  \BibitemOpen
  \bibfield  {author} {\bibinfo {author} {\bibfnamefont {J.}~\bibnamefont
  {Bensa}}\ and\ \bibinfo {author} {\bibfnamefont {M.}~\bibnamefont
  {{\v{Z}}nidari{\v{c}}}},\ }\bibfield  {title} {\bibinfo {title} {Purity decay
  rate in random circuits with different configurations of gates},\ }\href
  {https://doi.org/10.1103/physreva.107.022604} {\bibfield  {journal} {\bibinfo
   {journal} {Phys. Rev. A}\ }\textbf {\bibinfo {volume} {107}},\ \bibinfo
  {pages} {022604} (\bibinfo {year} {2023})},\ \Eprint
  {https://arxiv.org/abs/2211.13565} {arXiv:2211.13565 [quant-ph]} \BibitemShut
  {NoStop}%
\bibitem [{\citenamefont {Svetlichnyy}\ \emph {et~al.}()\citenamefont
  {Svetlichnyy}, \citenamefont {Mittal},\ and\ \citenamefont
  {Kennedy}}]{Svetlichnyy2022}%
  \BibitemOpen
  \bibfield  {author} {\bibinfo {author} {\bibfnamefont {P.}~\bibnamefont
  {Svetlichnyy}}, \bibinfo {author} {\bibfnamefont {S.}~\bibnamefont
  {Mittal}},\ and\ \bibinfo {author} {\bibfnamefont {T.~A.~B.}\ \bibnamefont
  {Kennedy}},\ }\href@noop {} {\bibinfo {title} {Matrix product states and the
  decay of quantum conditional mutual information}},\ \Eprint
  {https://arxiv.org/abs/2211.06794} {arXiv:2211.06794 [quant-ph]} \BibitemShut
  {NoStop}%
\bibitem [{\citenamefont {Haag}\ \emph {et~al.}(2023)\citenamefont {Haag},
  \citenamefont {Baccari},\ and\ \citenamefont {Styliaris}}]{Haag2023}%
  \BibitemOpen
  \bibfield  {author} {\bibinfo {author} {\bibfnamefont {D.}~\bibnamefont
  {Haag}}, \bibinfo {author} {\bibfnamefont {F.}~\bibnamefont {Baccari}},\ and\
  \bibinfo {author} {\bibfnamefont {G.}~\bibnamefont {Styliaris}},\ }\bibfield
  {title} {\bibinfo {title} {Typical {C}orrelation {L}ength of {S}equentially
  {G}enerated {T}ensor {N}etwork {S}tates},\ }\href
  {https://doi.org/10.1103/prxquantum.4.030330} {\bibfield  {journal} {\bibinfo
   {journal} {PRX Quantum}\ }\textbf {\bibinfo {volume} {4}},\ \bibinfo {pages}
  {030330} (\bibinfo {year} {2023})},\ \Eprint
  {https://arxiv.org/abs/2301.04624} {arXiv:2301.04624 [quant-ph]} \BibitemShut
  {NoStop}%
\bibitem [{\citenamefont {P{\'{e}}rez-Garc{\'{\i}}a}\ \emph
  {et~al.}(2007)\citenamefont {P{\'{e}}rez-Garc{\'{\i}}a}, \citenamefont
  {Verstraete}, \citenamefont {Wolf},\ and\ \citenamefont
  {Cirac}}]{PerezGarcia2007}%
  \BibitemOpen
  \bibfield  {author} {\bibinfo {author} {\bibfnamefont {D.}~\bibnamefont
  {P{\'{e}}rez-Garc{\'{\i}}a}}, \bibinfo {author} {\bibfnamefont
  {F.}~\bibnamefont {Verstraete}}, \bibinfo {author} {\bibfnamefont {M.~M.}\
  \bibnamefont {Wolf}},\ and\ \bibinfo {author} {\bibfnamefont {J.~I.}\
  \bibnamefont {Cirac}},\ }\bibfield  {title} {\bibinfo {title} {Matrix product
  state representations},\ }\href {https://doi.org/10.26421/qic7.5-6-1}
  {\bibfield  {journal} {\bibinfo  {journal} {Quant. Inf. Comp.}\ }\textbf
  {\bibinfo {volume} {7}},\ \bibinfo {pages} {401} (\bibinfo {year} {2007})},\
  \Eprint {https://arxiv.org/abs/quant-ph/0608197} {arXiv:quant-ph/0608197}
  \BibitemShut {NoStop}%
\bibitem [{\citenamefont {Cirac}\ \emph {et~al.}(2017)\citenamefont {Cirac},
  \citenamefont {Perez-Garcia}, \citenamefont {Schuch},\ and\ \citenamefont
  {Verstraete}}]{Cirac2017}%
  \BibitemOpen
  \bibfield  {author} {\bibinfo {author} {\bibfnamefont {J.~I.}\ \bibnamefont
  {Cirac}}, \bibinfo {author} {\bibfnamefont {D.}~\bibnamefont {Perez-Garcia}},
  \bibinfo {author} {\bibfnamefont {N.}~\bibnamefont {Schuch}},\ and\ \bibinfo
  {author} {\bibfnamefont {F.}~\bibnamefont {Verstraete}},\ }\bibfield  {title}
  {\bibinfo {title} {Matrix product unitaries: structure, symmetries, and
  topological invariants},\ }\href {https://doi.org/10.1088/1742-5468/aa7e55}
  {\bibfield  {journal} {\bibinfo  {journal} {J. Stat. Mech.}\ }\textbf
  {\bibinfo {volume} {2017}},\ \bibinfo {pages} {083105} (\bibinfo {year}
  {2017})},\ \Eprint {https://arxiv.org/abs/1703.09188} {arXiv:1703.09188
  [cond-mat.str-el]} \BibitemShut {NoStop}%
\bibitem [{\citenamefont {Collins}\ and\ \citenamefont
  {{\'{S}}niady}(2006)}]{Collins2006}%
  \BibitemOpen
  \bibfield  {author} {\bibinfo {author} {\bibfnamefont {B.}~\bibnamefont
  {Collins}}\ and\ \bibinfo {author} {\bibfnamefont {P.}~\bibnamefont
  {{\'{S}}niady}},\ }\bibfield  {title} {\bibinfo {title} {Integration with
  {R}espect to the {H}aar {M}easure on {U}nitary, {O}rthogonal and {S}ymplectic
  {G}roup},\ }\href {https://doi.org/10.1007/s00220-006-1554-3} {\bibfield
  {journal} {\bibinfo  {journal} {Comm. Math. Phys.}\ }\textbf {\bibinfo
  {volume} {264}},\ \bibinfo {pages} {773} (\bibinfo {year} {2006})},\ \Eprint
  {https://arxiv.org/abs/math-ph/0402073} {arXiv:math-ph/0402073} \BibitemShut
  {NoStop}%
\bibitem [{\citenamefont {Roberts}\ and\ \citenamefont
  {Yoshida}(2017)}]{Roberts2017}%
  \BibitemOpen
  \bibfield  {author} {\bibinfo {author} {\bibfnamefont {D.~A.}\ \bibnamefont
  {Roberts}}\ and\ \bibinfo {author} {\bibfnamefont {B.}~\bibnamefont
  {Yoshida}},\ }\bibfield  {title} {\bibinfo {title} {Chaos and complexity by
  design},\ }\href {https://doi.org/10.1007/JHEP04(2017)121} {\bibfield
  {journal} {\bibinfo  {journal} {J. High Energ. Phys.}\ }\textbf {\bibinfo
  {volume} {2017}},\ \bibinfo {pages} {121}},\ \Eprint
  {https://arxiv.org/abs/1610.04903} {arXiv:1610.04903 [quant-ph]} \BibitemShut
  {NoStop}%
\bibitem [{\citenamefont {Brand{\~{a}}o}\ \emph {et~al.}(2021)\citenamefont
  {Brand{\~{a}}o}, \citenamefont {Chemissany}, \citenamefont {Hunter-Jones},
  \citenamefont {Kueng},\ and\ \citenamefont {Preskill}}]{Brandao2021}%
  \BibitemOpen
  \bibfield  {author} {\bibinfo {author} {\bibfnamefont {F.~G. S.~L.}\
  \bibnamefont {Brand{\~{a}}o}}, \bibinfo {author} {\bibfnamefont
  {W.}~\bibnamefont {Chemissany}}, \bibinfo {author} {\bibfnamefont
  {N.}~\bibnamefont {Hunter-Jones}}, \bibinfo {author} {\bibfnamefont
  {R.}~\bibnamefont {Kueng}},\ and\ \bibinfo {author} {\bibfnamefont
  {J.}~\bibnamefont {Preskill}},\ }\bibfield  {title} {\bibinfo {title} {Models
  of {Q}uantum {C}omplexity {G}rowth},\ }\href
  {https://doi.org/10.1103/prxquantum.2.030316} {\bibfield  {journal} {\bibinfo
   {journal} {PRX Quantum}\ }\textbf {\bibinfo {volume} {2}},\ \bibinfo {pages}
  {030316} (\bibinfo {year} {2021})},\ \Eprint
  {https://arxiv.org/abs/1912.04297} {arXiv:1912.04297 [hep-th]} \BibitemShut
  {NoStop}%
\bibitem [{\citenamefont {Collins}(2003)}]{Collins2003}%
  \BibitemOpen
  \bibfield  {author} {\bibinfo {author} {\bibfnamefont {B.}~\bibnamefont
  {Collins}},\ }\bibfield  {title} {\bibinfo {title} {Moments and {C}umulants
  of {P}olynomial {R}andom {V}ariables on {U}nitary {G}roups, the
  {I}tzykson-{Z}uber {I}ntegral, and {F}ree {P}robability},\ }\href
  {https://doi.org/10.1155/S107379280320917X} {\bibfield  {journal} {\bibinfo
  {journal} {Int. Math. Res. Not.}\ }\textbf {\bibinfo {volume} {2003}},\
  \bibinfo {pages} {953} (\bibinfo {year} {2003})},\ \Eprint
  {https://arxiv.org/abs/math-ph/0205010} {arXiv:math-ph/0205010} \BibitemShut
  {NoStop}%
\bibitem [{Note1()}]{Note1}%
  \BibitemOpen
  \bibinfo {note} {Although the Weingarten matrix \( W = G^{- 1} \) exits only
  if \( k \leq q \)~\cite {Collins2021}, the Weingarten function can easily be
  extended to \( k > q \)~\cite {Collins2006}.}\BibitemShut {Stop}%
\bibitem [{\citenamefont {Watrous}(2018)}]{Watrous2018}%
  \BibitemOpen
  \bibfield  {author} {\bibinfo {author} {\bibfnamefont {J.}~\bibnamefont
  {Watrous}},\ }\href@noop {} {\emph {\bibinfo {title} {The Theory of Quantum
  Information}}}\ (\bibinfo  {publisher} {Cambridge University Press},\
  \bibinfo {address} {Cambridge},\ \bibinfo {year} {2018})\BibitemShut
  {NoStop}%
\bibitem [{\citenamefont {Haag}\ \emph {et~al.}()\citenamefont {Haag},
  \citenamefont {Milbradt},\ and\ \citenamefont {Mendl}}]{GitHub}%
  \BibitemOpen
  \bibfield  {author} {\bibinfo {author} {\bibfnamefont {D.}~\bibnamefont
  {Haag}}, \bibinfo {author} {\bibfnamefont {R.~M.}\ \bibnamefont {Milbradt}},\
  and\ \bibinfo {author} {\bibfnamefont {C.~B.}\ \bibnamefont {Mendl}},\ }\href
  {https://github.com/denialhaag/solvable} {\bibinfo {title} {solvable {GitHub}
  repository}}\BibitemShut {NoStop}%
\bibitem [{\citenamefont {Fukuda}\ \emph {et~al.}(2019)\citenamefont {Fukuda},
  \citenamefont {K{\"{o}}nig},\ and\ \citenamefont {Nechita}}]{Fukuda2019}%
  \BibitemOpen
  \bibfield  {author} {\bibinfo {author} {\bibfnamefont {M.}~\bibnamefont
  {Fukuda}}, \bibinfo {author} {\bibfnamefont {R.}~\bibnamefont
  {K{\"{o}}nig}},\ and\ \bibinfo {author} {\bibfnamefont {I.}~\bibnamefont
  {Nechita}},\ }\bibfield  {title} {\bibinfo {title} {{RTNI}---{A} symbolic
  integrator for {H}aar-random tensor networks},\ }\href
  {https://doi.org/10.1088/1751-8121/ab434b} {\bibfield  {journal} {\bibinfo
  {journal} {J. Phys. A: Math. Theor.}\ }\textbf {\bibinfo {volume} {52}},\
  \bibinfo {pages} {425303} (\bibinfo {year} {2019})},\ \Eprint
  {https://arxiv.org/abs/1902.08539} {1902.08539 [quant-ph]} \BibitemShut
  {NoStop}%
\bibitem [{\citenamefont {Collins}\ \emph {et~al.}()\citenamefont {Collins},
  \citenamefont {Matsumoto},\ and\ \citenamefont {Novak}}]{Collins2021}%
  \BibitemOpen
  \bibfield  {author} {\bibinfo {author} {\bibfnamefont {B.}~\bibnamefont
  {Collins}}, \bibinfo {author} {\bibfnamefont {S.}~\bibnamefont {Matsumoto}},\
  and\ \bibinfo {author} {\bibfnamefont {J.}~\bibnamefont {Novak}},\
  }\href@noop {} {\bibinfo {title} {The {W}eingarten {C}alculus}},\ \Eprint
  {https://arxiv.org/abs/2109.14890} {arXiv:2109.14890 [math-ph]} \BibitemShut
  {NoStop}%
\bibitem [{\citenamefont {Wolf}(2012)}]{Wolf2012}%
  \BibitemOpen
  \bibfield  {author} {\bibinfo {author} {\bibfnamefont {M.~M.}\ \bibnamefont
  {Wolf}},\ }\href {https://mediatum.ub.tum.de/doc/1701036/document.pdf}
  {\bibinfo {title} {Quantum {C}hannels {\&} {O}perations}},\ \bibinfo
  {howpublished} {Lecture notes} (\bibinfo {year} {2012})\BibitemShut {NoStop}%
\bibitem [{\citenamefont {Sanz}\ \emph {et~al.}(2010)\citenamefont {Sanz},
  \citenamefont {Perez-Garcia}, \citenamefont {Wolf},\ and\ \citenamefont
  {Cirac}}]{Sanz2010}%
  \BibitemOpen
  \bibfield  {author} {\bibinfo {author} {\bibfnamefont {M.}~\bibnamefont
  {Sanz}}, \bibinfo {author} {\bibfnamefont {D.}~\bibnamefont {Perez-Garcia}},
  \bibinfo {author} {\bibfnamefont {M.~M.}\ \bibnamefont {Wolf}},\ and\
  \bibinfo {author} {\bibfnamefont {J.~I.}\ \bibnamefont {Cirac}},\ }\bibfield
  {title} {\bibinfo {title} {A {Q}uantum {V}ersion of {W}ielandt's
  {I}nequality},\ }\href {https://doi.org/10.1109/tit.2010.2054552} {\bibfield
  {journal} {\bibinfo  {journal} {{IEEE} Trans. Inf. Theory}\ }\textbf
  {\bibinfo {volume} {56}},\ \bibinfo {pages} {4668} (\bibinfo {year}
  {2010})},\ \Eprint {https://arxiv.org/abs/0909.5347} {arXiv:0909.5347
  [quant-ph]} \BibitemShut {NoStop}%
\end{thebibliography}%

\onecolumngrid

\appendix

\section{Solvable MPS}

\subsection{Fixed points} \label{app:invariant_mps_points}

In this Appendix, we show that a solvable MPS \( \ket{\psi} \) as defined by Eq.~\eqref{eq:invariant_mps} is normalized. In the process, we will find that its transfer matrix \( E \) [see Eq.~\eqref{eq:invariant_mps_transfer_matrix}] has a unique largest eigenvalue \( \lambda = 1 \) and that the corresponding unique left and right fixed points are given by Eq.~\eqref{eq:mps_fixed_point}.

Corresponding to a completely positive trace-preserving map, the largest eigenvalue of the transfer matrix \begin{align}
    E = \frac{1}{d} \sum_{i_1, i_2 = 1}^d U_{i_1, i_2} \otimes \overline{U_{i_1, i_2}} = \frac{1}{d} \vcenter{\hbox{\includegraphics{invariant_mps_transfer_matrix}}}
\end{align} is \( \lambda = 1 \)~\cite{Wolf2012}. Injectivity then implies that said largest eigenvalue is unique~\cite{Sanz2010}. Given the unitarity of \( U \), it is evident that the corresponding unique left and right fixed points are, respectively, given by \begin{align}
    \bra{L} = \frac{1}{\sqrt{D}} \sum_{i = 1}^D \bra{i} \otimes \bra{i} = \frac{1}{\sqrt{D}} \vcenter{\hbox{\includegraphics{mps_fixed_point_left}}} \qquad \text{and} \qquad \ket{R} = \frac{1}{\sqrt{D}} \sum_{i = 1}^D \ket{i} \otimes \ket{i} = \frac{1}{\sqrt{D}} \vcenter{\hbox{\includegraphics{mps_fixed_point_right}}},
\end{align} where we have chosen the factors of \( 1 / \sqrt{D} \) so that \( \braket{L}{R} = 1 \).

Using the identity \( \braket{\psi}{\psi} = \tr \bigl ( E^n \bigr ) \), the norm of \( \ket{\psi} \) is related to its transfer matrix. As we consider the thermodynamic limit of \( n \to \infty \), the uniqueness of the largest eigenvalue \( \lambda = 1 \) is enough to conclude that \begin{align} \label{eq:invariant_mps_limit}
    E^n & = \ketbra{R}{L} = \frac{1}{D} \vcenter{\hbox{\includegraphics{mps_fixed_point_right}}} \vcenter{\hbox{\includegraphics{mps_fixed_point_left}}}.
\end{align} It is then evident that \begin{align}
    \braket{\psi}{\psi} = \tr \bigl ( E^n \bigr ) = \tr \bigl ( \ketbra{R}{L} \bigr ) = \braket{L}{R} = 1.
\end{align}

\subsection{Correlations} \label{app:invariant_mps_correlation_function}

In this Appendix, we discuss the two-point equal-time correlation function \( C^{\alpha \beta} (x, r, t) \) [see Eqs.\eqref{eq:mps_correlation_function} and \eqref{eq:adrian}]. For reference, for \( \alpha \neq 0 \) and \( \beta \neq 0 \), \begin{align}
    C^{\alpha \beta} (x, r, t) = \fv{\psi (t)}{a_x^\alpha a_{x + r}^\beta}{\psi (t)} = \begin{cases}
        0 & x \ \text{even} \\
        0 & r \ \text{even} \\
        0 & r < 4 t + 1 \\
        D_1^{\alpha \beta} (x, r, t) & r = 4 t + 1 \\
        D_2^{\alpha \beta} (x, r, t) & r > 4 t + 1
    \end{cases}.
\end{align}

\subsubsection{\texorpdfstring{Case of \( r > 4 t + 1 \)}{General case}}

Let us consider the correlation function sketched in Fig.~\ref{fig:invariant_mps_correlation_function} as an example for the case of \( r > 4 t + 1 \). Let \( x \) be an odd integer, \( r = 9 \), and \( t = 1 \). Then, \begin{align}
    D_2^{\alpha \beta} (x, r, t) = \frac{1}{d^n} \vcenter{\hbox{\includegraphics{invariant_mps_correlation_function_lhs}}}.
\end{align}

Using the temporal unitarity of \( \mathcal{U} \) [see Eq.~\eqref{eq:temporal_unitarity}], the expression immediately reduces to \begin{align} \label{eq:invariant_mps_correlation_function_mid_1}
    D_2^{\alpha \beta} (x, r, t) = \frac{1}{d^n} \vcenter{\hbox{\includegraphics{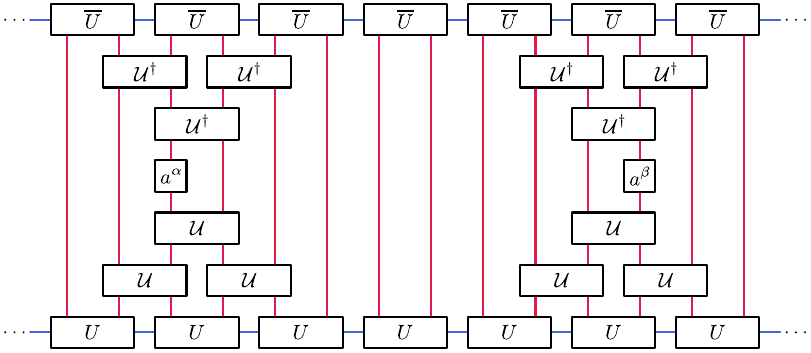}}}.
\end{align}

The \( n - 7 \) pairs of sites not drawn in Eq.~\eqref{eq:invariant_mps_correlation_function_mid_1} experience trivial action, implying that they are given by the transfer matrix \( E \) [see Eq.~\eqref{eq:invariant_mps_transfer_matrix}]. With \begin{align}
    G = \frac{1}{d^7} \vcenter{\hbox{\includegraphics{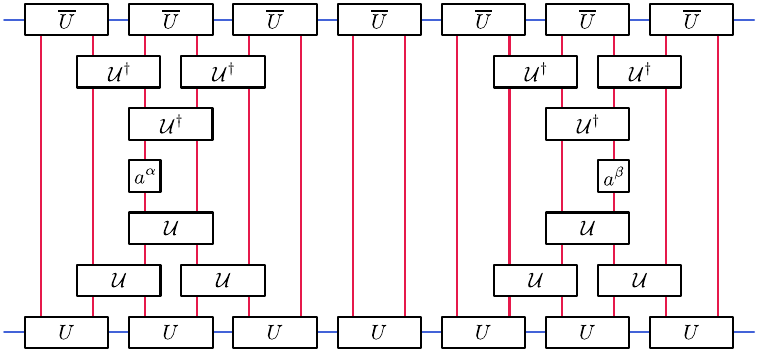}}},
\end{align} the correlation function is thus given by \begin{align}
    D_2^{\alpha \beta} (x, r, t) = \tr \left ( G E^{n - 7} \right ).
\end{align} In the thermodynamic limit of \( n \to \infty \), \begin{align}
    D_2^{\alpha \beta} (x, r, t) = \tr \bigl ( G \ketbra{R}{L} \bigr ) = \fv{L}{G}{R}.
\end{align} That is, \begin{align}
    D_2^{\alpha \beta} (x, r, t) = \frac{1}{d^7 D} \vcenter{\hbox{\includegraphics{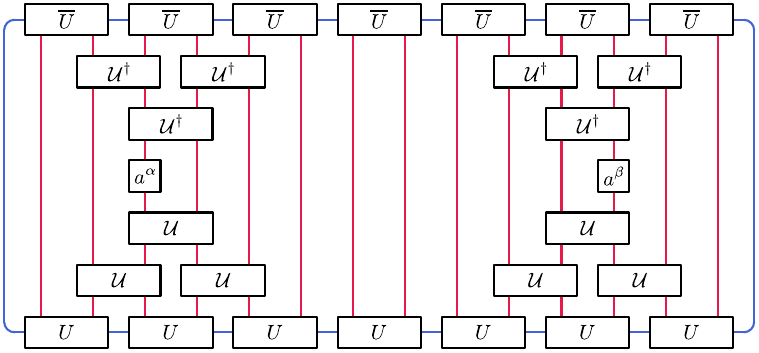}}}.
\end{align}

Using the unitarity of \( U \), the expression then reduces to \begin{align}
    D_2^{\alpha \beta} (x, r, t) = \frac{1}{d^7 D} \vcenter{\hbox{\includegraphics{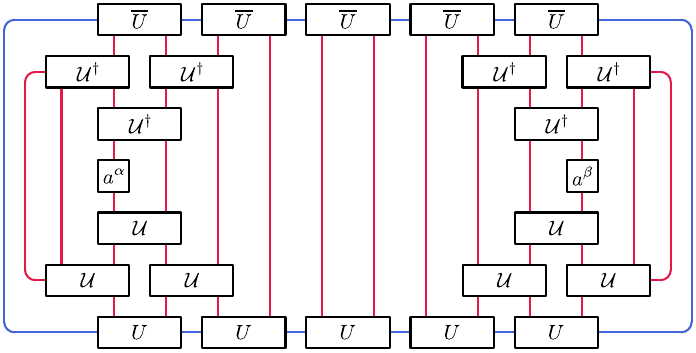}}},
\end{align} and using the spatial unitarity of \( \mathcal{U} \) [see Eq.~\eqref{eq:spatial_unitarity}] and again the unitarity of \( U \) leads us to \begin{align}
    D_2^{\alpha \beta} (x, r, t) = \frac{1}{d^7 D} \vcenter{\hbox{\includegraphics{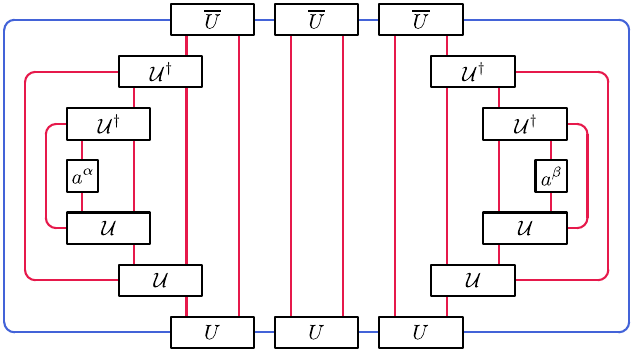}}}.
\end{align}

Motivated by Ref.~\cite{Bertini2019}, we define the maps \( \mathcal{M}_+ \) and \( \mathcal{M}_- \) defined in Eq.~\eqref{eq:m}. For reference, \begin{subequations}
    \begin{align}
        \mathcal{M}_+ (a) = \frac{1}{d} \tr_1 \left [ \mathcal{U}^\dagger (a \otimes I) \mathcal{U} \right ] = \frac{1}{d} \vcenter{\hbox{\includegraphics{mps_m_plus}}}
    \end{align} and \begin{align}
        \mathcal{M}_- (a) = \frac{1}{d} \tr_2 \left [ \mathcal{U}^\dagger (I \otimes a) \mathcal{U} \right ] = \frac{1}{d} \vcenter{\hbox{\includegraphics{mps_m_minus}}}.
    \end{align}
\end{subequations}

Finally, we define \( A = \mathcal{M}_+^{2 t} \left ( a^\alpha \right ) \) and \( B = \mathcal{M}_-^{2 t} \left ( a^\beta \right ) \) so that \begin{align}
    D_2^{\alpha \beta} (x, r, t) = \frac{1}{d^3 D} \vcenter{\hbox{\includegraphics{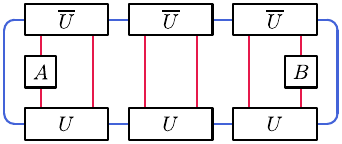}}},
\end{align} and \( V = U \otimes \overline{U} \) so that \begin{align} \label{eq:invariant_mps_correlation_function_rhs}
    D_2^{\alpha \beta} (x, r, t) = \frac{1}{d^3 D} \raisebox{-7pt}{\includegraphics{invariant_mps_correlation_function_rhs}}.
\end{align}

Note that the final expression coincides with Eq.~\eqref{eq:fabian}.

\subsubsection{\texorpdfstring{Case of \( r = 4 t + 1 \)}{Special case}}

As an example for the special case of \( r = 4 t + 1 \), let \( x \) be an odd integer, \( r = 5 \), and \( t = 1 \). Following the same steps as in the case of \( r > 4 t + 1 \), it is straightforward to confirm that \begin{align}
    D_1^{\alpha \beta} (x, 5, 1) & = \frac{1}{d^n} \vcenter{\hbox{\includegraphics{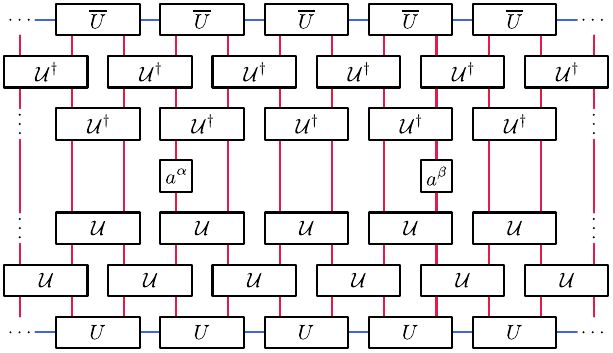}}} \\
    & = \frac{1}{d D} \raisebox{-7pt}{\includegraphics{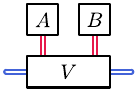}}.
\end{align}

\subsubsection{\texorpdfstring{Cases of even \( x \) and even \( r \)}{Vanishing case}}

If \( x \) were an even integer in either of the cases discussed above, one would obtain a factor of \( \tr (A) \). However, using the spatial unitarity of \( \mathcal{U} \) [see Eq.~\eqref{eq:spatial_unitarity}], it holds that \( \tr (A) = \tr \left ( a^\alpha \right ) \). \( C^{\alpha \beta} (x, r, t) \) would then vanish, except for the trivial case of \( \alpha = 0 \). A similar argument holds for the case of even \( r \).

\section{Disordered solvable MPS} \label{app:mps}

We define a disordered solvable MPS \( \ket{\psi} \) as a \( 2 v \)-site shift-invariant MPS that is given by \( w \in \mathbb{N} \) blocks of \( v \in \mathbb{N} \) consecutive unitary matrices \( U^{(j)} \in \mathrm{U} (d D) \) with \( j \in \{ 1, \dots, v \} \) such that \( v w = n \). That is, \begin{align} \label{eq:mps}
    \ket{\psi} & = \frac{1}{\sqrt{d^n}} \sum_{i_1, \dots, i_{2 v w}} \tr \left ( U_{i_1, i_2}^{(1)} U_{i_3, i_4}^{(2)} \cdots U_{i_{2 v - 1}, i_{2 v}}^{(v)} U_{i_{2 v + 1}, i_{2 v + 2}}^{(1)} \cdots U_{i_{2 v w - 1}, i_{2 v w}}^{(v)} \right ) \ket{i_1 \cdots i_{2 v w}} \\
    & = \frac{1}{\sqrt{d^n}} \raisebox{-7pt}{\includegraphics{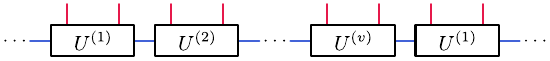}}.
\end{align} Similarly to the step between Eqs.~\eqref{eq:invariant_mps_ab} and \eqref{eq:invariant_mps_w}, one can define a matrix \( Q \in \End \bigl ( \left ( \mathbb{C}^{d} \right )^{\otimes v} \otimes \mathbb{C}^D \bigr ) \) such that \begin{align}
    \ket{\psi} = \frac{1}{\sqrt{d^n}} \sum_{i_1, \dots, i_{2 v w}} \tr \left ( Q_{i_1, \dots, i_{2 v}} \cdots Q_{i_{2 v w - 2 v + 1}, \dots, i_{2 v w}} \right ) \ket{i_1 \cdots i_{2 v w}}.
\end{align}

A disordered solvable MPS thus corresponds to a translation-invariant MPS with physical dimension \( d^{2 v} \). As we discuss in App~\ref{app:mps_points}, this symmetry is necessary to prove the existence of unique left and right fixed points.

\subsection{Fixed points} \label{app:mps_points}

In this Appendix, we show that the unique left and right fixed points of a disordered solvable MPS \( \ket{\psi} \) as defined by Eq.~\eqref{eq:mps} are, respectively, given by \begin{align} \label{eq:mps_fixed_points}
    \bra{L} = \frac{1}{\sqrt{D}} \sum_{i = 1}^D \bra{i} \otimes \bra{i} = \frac{1}{\sqrt{D}} \vcenter{\hbox{\includegraphics{mps_fixed_point_left}}} \qquad \text{and} \qquad \ket{R} = \frac{1}{\sqrt{D}} \sum_{i = 1}^D \ket{i} \otimes \ket{i} = \frac{1}{\sqrt{D}} \vcenter{\hbox{\includegraphics{mps_fixed_point_right}}},
\end{align} as in the nondisordered case (see App.~\ref{app:invariant_mps_points}).

Given App.~\ref{app:invariant_mps_points}, the proof is all but straightforward. Under the assumption that \( \ket{\psi} \) is injective in the thermodynamic limit of \( w \to \infty \), the largest eigenvalue of its transfer matrix \begin{align} \label{eq:mps_transfer_matrix}
    F = \frac{1}{d^v} \sum_{i_1, \dots, i_{2 v} = 1}^d Q_{i_1, \dots, i_{2 v}} \otimes \overline{Q_{i_1, \dots, i_{2 v}}} = \frac{1}{d^v} \vcenter{\hbox{\includegraphics{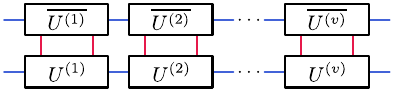}}}
\end{align} is \( \lambda = 1 \) and unique~\cite{Wolf2012, Sanz2010}. Given the unitarity of \( U^{(j)} \) with \( j \in \{ 1, \dots, v \} \), it is then evident that the corresponding unique left and right fixed points are given by Eq.~\eqref{eq:mps_fixed_points}.

Furthermore, it is evident that \begin{align}
    \braket{\psi}{\psi} = \tr \bigl( F^w \bigr ) = \tr \bigl ( \ketbra{R}{L} \bigr) = \braket{L}{R} = 1.
\end{align}

\subsection{Correlations} \label{app:mps_correlation_function}

In this Appendix, we discuss that the procedure of simplifying \( D_1^{\alpha \beta} (x, r, t) \) and \( D_2^{\alpha \beta} (x, r, t) \) is analogous to that of the nondisordered case (see App.~\ref{app:invariant_mps_correlation_function}).

Indeed, most of the arguments of App.~\ref{app:invariant_mps_correlation_function} hold exactly; only the definition of \( G \) has to be adapted slightly. With \begin{align}
    G = \frac{1}{d^v} \vcenter{\hbox{\includegraphics{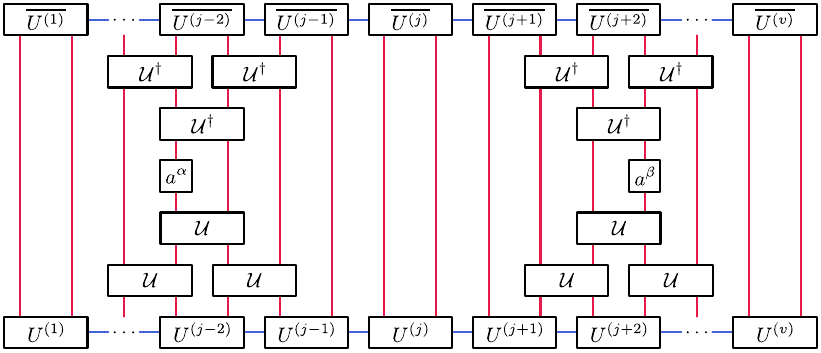}}},
\end{align} the correlation function is given by \begin{align}
    D_2^{\alpha \beta} (x, r, t) = \tr \left ( G E^{w - 1} \right ).
\end{align} In the thermodynamic limit of \( w \to \infty \), \begin{align}
    D_2^{\alpha \beta} (x, r, t) = \tr \bigl ( G \ketbra{R}{L} \bigr ) = \fv{L}{G}{R},
\end{align} which reduces to \begin{align}
    D_2^{\alpha \beta} (x, r, t) = \frac{1}{d^3 D} \raisebox{-7pt}{\includegraphics{mps_correlation_function_rhs}}.
\end{align}

\subsection{Computing averages} \label{app:mps_average}

From Eqs.~\eqref{eq:mps_r_1} and \eqref{eq:mps_r_2}, recall that the building block for computing averages with respect to the random disordered solvable MPS ensemble is given by \begin{align}
	\raisebox{-7pt}{\includegraphics{mps_building_block_lhs}} = \int \mathrm{d} U^{(j)} \, \raisebox{-14pt}{\includegraphics{mps_building_block_mid}} = \raisebox{-7pt}{\includegraphics{mps_building_block_rhs}}.
\end{align}

\subsubsection{\texorpdfstring{Case of \( r > 4 t + 1 \)}{General case}}

As stated in Eq.~\eqref{eq:sebastian}, the average of the \( k \)th moment of \( D_2^{\alpha \beta} (x, r, t) \) is given by \begin{align}
	\mathbb{E} \bigl [ D_2^{\alpha \beta} (x, r, t) \bigr ]^k = \frac{1}{\left ( d^{s + 2} D \right )^k} \vcenter{\hbox{\includegraphics{mps_average_rhs}}},
\end{align} where \( s \) is defined in Eq.~\eqref{eq:s}.

Using Eq.~\eqref{eq:permutation_contraction}, the entries of the matrix \( T \in \mathbb{R}^{k! \times k!} \) are given by \begin{align}
	\raisebox{-7pt}{\includegraphics{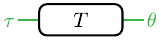}} & = \raisebox{-7pt}{\includegraphics{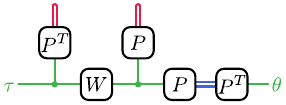}} \\
	& = \sum_{\sigma \in S_k} \Wg \left ( \sigma \tau^{- 1}, d D \right ) d^{\# (\tau)} d^{\# (\sigma)} D^{\# \left ( \sigma \theta^{- 1} \right )}, \label{eq:mps_t}
\end{align} those of the left vector \( \bra{T^{(\ell)}} \in \mathbb{R}^{k!} \) are given by \begin{align}
    \raisebox{-7pt}{\includegraphics{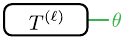}} & = \raisebox{-7pt}{\includegraphics{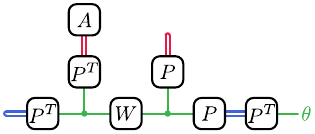}} \\
    & = \sum_{\sigma, \tau \in S_k} \Wg \left ( \sigma \tau^{- 1}, d D \right ) \tr \left [ \left (  P_\tau^{(d)} \right )^T A^{\otimes k} \right ] d^{\# (\sigma)} D^{\# (\tau)} D^{\# \left ( \sigma \theta^{- 1} \right )}, \label{eq:mps_t-l}
\end{align} and those of the right vector \( \ket{T^{(\mathrm{r})}} \in \mathbb{R}^{k!} \) are given by \begin{align}
    \raisebox{-7pt}{\includegraphics{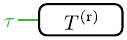}} & = \raisebox{-7pt}{\includegraphics{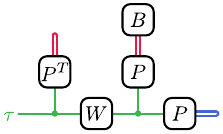}} \\
    & = \sum_{\sigma \in S_k} \Wg \left ( \sigma \tau^{- 1}, d D \right ) \tr \left ( P_\sigma^{(d)} B^{\otimes k} \right ) d^{\# (\tau)} D^{\# (\sigma)}. \label{eq:mps_t-r}
\end{align}

\subsubsection{\texorpdfstring{Case of \( r = 4 t + 1 \)}{Special case}}

As stated in Eq.~\eqref{eq:mps_special_average}, for arbitrary but compatible \( r \) and \( t \), the average of the \( k \)th moment of \( D_1^{\alpha \beta} (x, r, t) \) is given by \begin{align}
	\mathbb{E} \bigl [ D_1^{\alpha \beta} (x, r, t) \bigr ]^k & = \frac{1}{\left ( d D \right )^k} \raisebox{-7pt}{\includegraphics{mps_special_average_lhs}} = \frac{1}{\left ( d D \right )^k} \raisebox{-7pt}{\includegraphics{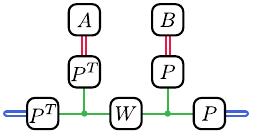}} \\
	& = \frac{1}{\left ( d D \right )^k} \sum_{\sigma, \tau \in S_k} \Wg \left ( \sigma \tau^{- 1}, d D \right ) \tr \left [ \left ( P_\tau^{(d)} \right )^T A^{\otimes k} \right ] \tr \left ( P_\sigma^{(d)} B^{\otimes k} \right ) D^{\# (\tau)} D^{\# (\sigma)}.
\end{align}

\subsection{Proof of Corollary~\ref{cor:mps_average_correlation_function}} \label{app:mps_average_correlation_function}

\mpsaveragecorrelationfunction*

\begin{proof}
    For \( k = 1 \), the Weingarten matrix \( W = G^{- 1} \) (see Sec.~\ref{sec:k-twirl}) has only one entry, namely \begin{align}
        \Wg (e, d D) = \frac{1}{d D}.
    \end{align} With that, one may confirm that \begin{align}
        \mathbb{E} D_1^{\alpha \beta} (x, r, t) = \mathbb{E} D_2^{\alpha \beta} (x, r, t) = \frac{1}{d^2} \raisebox{-18pt}{\includegraphics{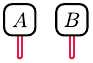}} = \frac{1}{d^2} \tr (A) \tr (B).
    \end{align} Using the spatial unitarity of \( \mathcal{U} \) [see Eq.~\eqref{eq:spatial_unitarity}], it holds that \begin{align}
        & \tr (A) = \tr \left [ \mathcal{M}_+^{2 t} \left ( a^\alpha \right ) \right ] = \tr \left ( a^\alpha \right ), \\
        & \tr (B) = \tr \left [ \mathcal{M}_-^{2 t} \left ( a^\beta \right ) \right ] = \tr \left ( a^\beta \right ).
    \end{align} The statement follows.
\end{proof}

\section{Disordered solvable PEPS} \label{app:peps}

We define a disordered solvable PEPS \( \ket{\psi} \) as a PEPS that is \( 2 v \)-site shift-invariant in \( x_1 \)-direction in that \begin{align} \label{eq:peps}
    \ket{\psi} & = \frac{1}{\sqrt{d^{m n} D^n}} \raisebox{-78pt}{\includegraphics{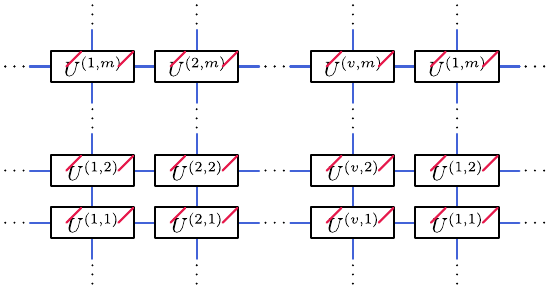}},
\end{align} where the tensors \( U^{(i, j)} \) with \( i \in \{ 1, \dots, v \} \) and \( j \in \{ 1, \dots, m \} \) are parameterized by unitary matrices \( \widetilde{U}^{(i, j)} \) and \( \widehat{U}^{(i, j)} \) in that \begin{align}
    \vcenter{\hbox{\includegraphics{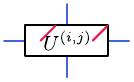}}} = \vcenter{\hbox{\includegraphics{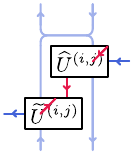}}}.
\end{align} Note that the tensors \( U^{(i, j)} \) satisfy \begin{align} \label{eq:peps_unitarity}
    \vcenter{\hbox{\includegraphics{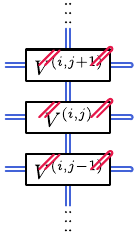}}} = \vcenter{\hbox{\includegraphics{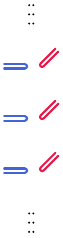}}}
\end{align} as well as \begin{subequations}
    \begin{align} \label{eq:peps_simplicity_1}
        \vcenter{\hbox{\includegraphics{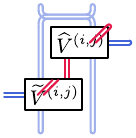}}} = \vcenter{\hbox{\includegraphics{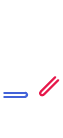}}}
    \end{align} and \begin{align} \label{eq:peps_simplicity_2}
        \vcenter{\hbox{\includegraphics{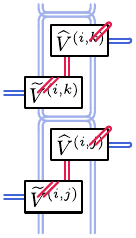}}} = \vcenter{\hbox{\includegraphics{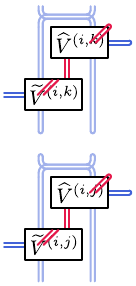}}},
    \end{align}
\end{subequations} in analogy to Eqs.~\eqref{eq:mpu_condition}, and \eqref{eq:invariant_peps_simplicity}.

It will prove beneficial to think in terms of MPS instead of PEPS. We achieve this by contracting the PEPS vertically and defining matrices \( W^{(i)} \in \End \bigl ( \left ( \mathbb{C}^{d} \right )^{\otimes m} \otimes \left ( \mathbb{C}^D \right )^{\otimes m} \bigr ) \) with \( i \in \{ i, \dots, v \} \) such that \begin{align}
    \ket{\psi} & = \frac{1}{\sqrt{d^{m n} D^n}} \sum_{i_1, \dots, i_{2 v w}} \tr \left ( W_{i_1, i_2}^{(1)} W_{i_3, i_4}^{(2)} \cdots W_{i_{2 v - 1}, i_{2 v}}^{(v)} W_{i_{2 v + 1}, i_{2 v + 2}}^{(1)} \cdots W_{i_{2 v w - 1}, i_{2 v w}}^{(v)} \right ) \ket{i_1 \cdots i_{2 v w}} \label{eq:peps_w} \\
    & = \frac{1}{\sqrt{d^{m n} D^n}} \raisebox{-7pt}{\includegraphics{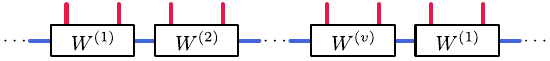}}.
\end{align} As in one dimension (see App.~\ref{app:mps}), one can then define a matrix \( Q \in \End \bigl ( \left ( \mathbb{C}^{d} \right )^{\otimes m v} \otimes \left ( \mathbb{C}^D \right )^{\otimes m} \bigr ) \) such that \begin{align}
    \ket{\psi} = \frac{1}{\sqrt{d^{m n} D^n}} \sum_{i_1, \dots, i_{2 v w}} \tr \left ( Q_{i_1, \dots, i_{2 v}} \cdots Q_{i_{2 v w - 2 v + 1}, \dots, i_{2 v w}} \right ) \ket{i_1 \cdots i_{2 v w}}.
\end{align}

A disordered solvable PEPS thus corresponds to a translation-invariant MPS with physical dimension \( d^{2 m v} \). As we discuss in App~\ref{app:peps_points}, this symmetry is necessary to prove the existence of unique left and right fixed points.

\subsection{Fixed points} \label{app:peps_points}

In this Appendix, we show that the unique left and right fixed points of the MPS \( \ket{\psi} \) defined in Eq.~\eqref{eq:peps_w} are, respectively, given by \begin{align} \label{eq:peps_fixed_points}
    \bra{L} = \frac{1}{\sqrt{D^m}} \sum_{i = 1}^{D^m} \bra{i} \otimes \bra{i} = \frac{1}{\sqrt{D^m}} \vcenter{\hbox{\includegraphics{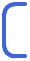}}} \qquad \text{and} \qquad \ket{R} = \frac{1}{\sqrt{D^m}} \sum_{i = 1}^{D^m} \ket{i} \otimes \ket{i} = \frac{1}{\sqrt{D^m}} \vcenter{\hbox{\includegraphics{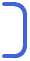}}},
\end{align} in correspondence to the one-dimensional case (see App.~\ref{app:mps_points}).

Given App.~\ref{app:invariant_mps_points}, the proof is all but straightforward. Under the assumption that the MPS \( \ket{\psi} \) defined in Eq.~\eqref{eq:peps_w} is injective in the thermodynamic limit of \( w \to \infty \), the largest eigenvalue of its transfer matrix \begin{align} \label{eq:peps_transfer_matrix}
    F = \frac{1}{d^{m v} D^v} \sum_{i_1, \dots, i_{2 v} = 1}^d Q_{i_1, \dots, i_{2 v}} \otimes \overline{Q_{i_1, \dots, i_{2 v}}} = \frac{1}{d^{m v} D^v} \vcenter{\hbox{\includegraphics{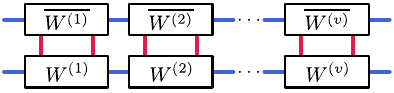}}}
\end{align} is \( \lambda = 1 \) and unique~\cite{Wolf2012, Sanz2010}. Because the tensors \( U^{(i, j)} \) with \( i \in \{ 1, \dots, v \} \) and \( j \in \{ 1, \dots, m \} \) satisfy Eq.~\eqref{eq:peps_unitarity}, the matrices \( W^{(i)} \) are unitary, implying that the corresponding unique left and right fixed points are given by Eq.~\eqref{eq:peps_fixed_points}.

It is evident that \begin{align}
    \braket{\psi}{\psi} = \tr \bigl( F^n \bigr ) = \tr \bigl ( \ketbra{R}{L} \bigr) = \braket{L}{R} = 1.
\end{align}

\subsection{Correlations} \label{app:peps_correlation_function}

In this Appendix, we discuss that the procedure of simplifying \( D_1^{\alpha \beta} (x, r, t) \) and \( D_2^{\alpha \beta} (x, r, t) \) is all but identical to in one dimension (see App.~\ref{app:mps_correlation_function}).

The underlying idea is to contract the PEPS vertically, giving rise to the matrices \( W^{(i)} \) with \( i \in \{ 1, \dots, v \} \) defined in Eq.~\eqref{eq:peps_w}. Because the tensors \( U^{(i, j)} \) with \( i \in \{ 1, \dots, v \} \) and \( j \in \{ 1, \dots, m \} \) satisfy Eq.~\eqref{eq:peps_unitarity}, the matrices \( W^{(i)} \) are unitary, implying that we can employ the exact same machinery as in App.~\ref{app:mps_correlation_function} to simplify \( D_1^{\alpha \beta} (x, r, t) \) and \( D_2^{\alpha \beta} (x, r, t) \) in the \( x_1 \)-direction.

For subsequently simplifying the expressions in the \( x_2 \)-direction, we rely on simplicity. The thermodynamic limit of \( m \to \infty \) guarantees that there exist tensors experiencing trivial action in the \( x_2 \)-direction. By Eq.~\eqref{eq:peps_simplicity_2}, we may thus cut the PEPS parallel to the \( x_1 \)-direction, leaving us with open boundary boundary conditions in the \( x_2 \)-direction. We may then use Eq.~\eqref{eq:peps_simplicity_1} to remove the tensors experiencing trivial action, leading us to \begin{align}
    D_2^{\alpha \beta} (x, r, t) = \frac{1}{d^{18} D^9} \vcenter{\hbox{\includegraphics{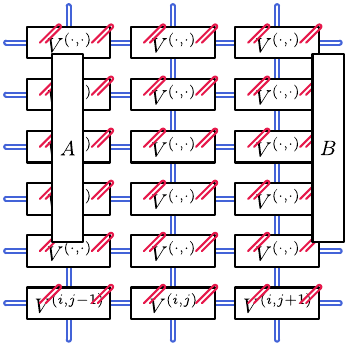}}}.
\end{align}

Eq.~\eqref{eq:peps_simplicity_2} is also the reason for the Heaviside step function in Eq.~\eqref{eq:simone}. Consider, for example, the case of odd \( x_1 \), even \( x_2 \), \( r_1 = 9 \), \( r_2 = 7 \), and \( t = 1 \), which is given by \begin{align}
    D_2^{\alpha \beta} (x, r, t) = \frac{1}{d^{36} D^{15}} \vcenter{\hbox{\includegraphics{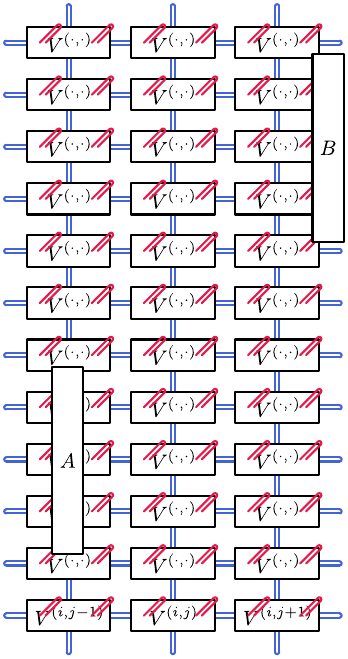}}}.
\end{align} By Eq.~\eqref{eq:peps_simplicity_2}, we may cut the PEPS parallel to the \( x_1 \)-direction, separating \( A \) from \( B \). This, however, leads to factors of \( \tr (A) \) and \( \tr (B) \). Using the spatial unitarity of \( \mathcal{U} \), it holds that \( \tr (A) = \tr \left ( a^\alpha \right ) \) and \( \tr (B) = \tr \left ( a^\beta \right ) \), implying that \( C^{\alpha \beta} (x, r, t) \) vanishes, except for the trivial case of \( \alpha = \beta = 0 \).

\subsection{Computing averages} \label{app:peps_average}
 
From Eqs.~\eqref{eq:peps_r_1} and \eqref{eq:peps_r_2}, recall that the building block for computing averages with respect to the random disordered solvable PEPS ensemble is given by \begin{align}
	\vcenter{\hbox{\includegraphics{peps_building_block_lhs}}} & = \int \mathrm{d} \widetilde{U}^{(j)} \mathrm{d} \widehat{U}^{(j)} \, \vcenter{\hbox{\includegraphics{peps_building_block_mid_1}}} = \int \mathrm{d} \widetilde{U}^{(j)} \mathrm{d} \widehat{U}^{(j)} \, \vcenter{\hbox{\includegraphics{peps_building_block_mid_2}}} \\
	& = \raisebox{-37pt}{\includegraphics{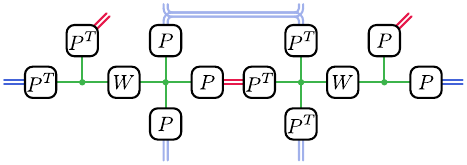}},
\end{align} where we have used that \begin{align}
    P_\sigma^{(d D)} = P_\sigma^{(d)} \otimes P_\sigma^{(D)}.
\end{align}

\subsubsection{\texorpdfstring{Case of \( r > 4 t + 1 \)}{General case}}

As stated in Eq.~\eqref{eq:peps_average}, for the case of odd \( x_1 \), even \( x_2 \), \( r_1 = 9 \), \( r_2 = 0 \) and \( t = 1 \), the average of the \( k \)th moment of \( D_2^{\alpha \beta} (x, r, t) \) is given by \begin{align}
    \mathbb{E} \bigl [ D_2^{\alpha \beta} (x, r, t) \bigr ]^k = \frac{1}{\left ( d^{18} D^9 \right )^k} \vcenter{\hbox{\includegraphics{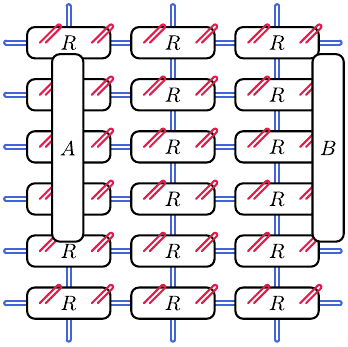}}}.
\end{align} By cutting permutation-valued legs instead of bond legs, we obtain \begin{align}
    \mathbb{E} \bigl [ D_2^{\alpha \beta} (x, r, t) \bigr ]^k = \frac{\hbox{\includegraphics{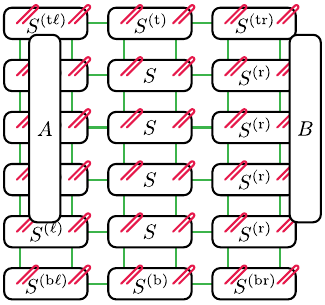}}}{\left ( d^{18} D^9 \right )^k} = \frac{\hbox{\includegraphics{peps_average_rhs}}}{\left ( d^{18} D^9 \right )^k}.
\end{align}

Using Eq.~\eqref{eq:permutation_contraction}, the entries of the tensor \( T \in \mathbb{R}^{k! \times k! \times k! \times k! \times k! \times k!} \) in the bulk are given by \begin{align}
    \raisebox{-28pt}{\includegraphics{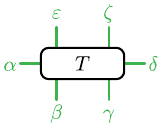}} & = \raisebox{-28pt}{\includegraphics{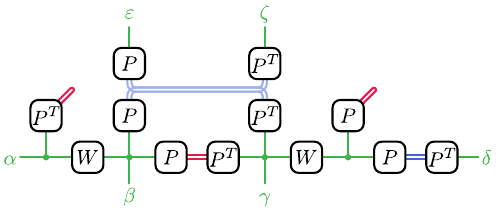}} \\
    & = \sum_{\sigma \in S_k} \Wg \left ( \beta \alpha^{- 1}, d D \right ) \Wg \left ( \sigma \gamma^{- 1}, d D \right ) \nonumber \\
    & \hphantom{{} = \sum_{\sigma \in S_k}} {} \times d^{\# (\alpha)} d^{\# (\sigma)} d^{\# \left ( \beta \gamma^{- 1} \right )} \nonumber \\
    & \hphantom{{} = \sum_{\sigma \in S_k}} {} \times D^{\# \left ( \sigma \delta^{- 1} \right )} D^{\# \left ( \beta \zeta^{- 1} \right ) / 2} D^{\# \left ( \varepsilon \gamma^{- 1} \right ) / 2},
\end{align} those of the tensor \( T^{(\ell)} \in \mathbb{R}^{k! \times k! \times k! \times k! \times k!} \) on the left boundary are given by \begin{align}
    \raisebox{-28pt}{\includegraphics{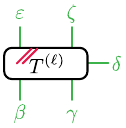}} & = \raisebox{-28pt}{\includegraphics{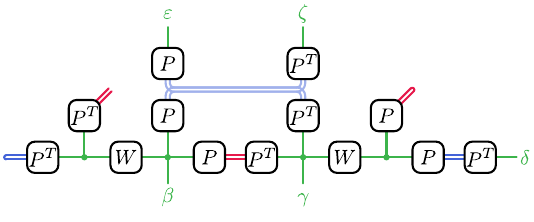}} \\
    & = \sum_{\alpha, \sigma \in S_k} \Wg \left ( \beta \alpha^{- 1}, d D \right ) \Wg \left ( \sigma \gamma^{- 1}, d D \right ) \nonumber \\
    & \hphantom{{} = \sum_{\alpha, \sigma \in S_k}} {} \times \left ( P_\alpha^{(d)} \right )^T \nonumber \\
    & \hphantom{{} = \sum_{\alpha, \sigma \in S_k}} {} \times d^{\# (\sigma)} d^{\# \left ( \beta \gamma^{- 1} \right )} \nonumber \\
    & \hphantom{{} = \sum_{\alpha, \sigma \in S_k}} {} \times D^{\# (\alpha)} D^{\# \left ( \sigma \delta^{- 1} \right )} D^{\# \left ( \beta \zeta^{- 1} \right ) / 2} D^{\# \left ( \varepsilon \gamma^{- 1} \right ) / 2},
\end{align} those of the tensor \( T^{(\mathrm{b} \ell)} \in \mathbb{R}^{k! \times k! \times k!} \) in the bottom-left corner are given by \begin{align}
    \raisebox{-38pt}{\includegraphics{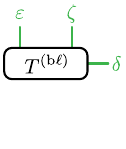}} & = \raisebox{-38pt}{\includegraphics{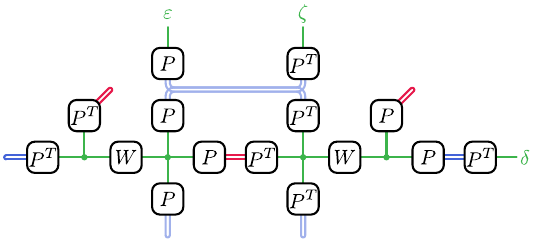}} \\
    & = \sum_{\alpha, \beta, \gamma, \sigma \in S_k} \Wg \left ( \beta \alpha^{- 1}, d D \right ) \Wg \left ( \sigma \gamma^{- 1}, d D \right ) \nonumber \\
    & \hphantom{{} = \sum_{\alpha, \beta, \gamma, \sigma \in S_k}} {} \times d^{\# (\alpha)} d^{\# (\sigma)} d^{\# \left ( \beta \gamma^{- 1} \right )} \nonumber \\
    & \hphantom{{} = \sum_{\alpha, \beta, \gamma, \sigma \in S_k}} {} \times D^{\# (\alpha)} D^{\# \left ( \sigma \delta^{- 1} \right )} D^{\# (\beta) / 2} D^{\# (\gamma) / 2} D^{\# \left ( \beta \zeta^{- 1} \right ) / 2} D^{\# \left ( \varepsilon \gamma^{- 1} \right ) / 2},
\end{align} those of the tensor \( T^{(\mathrm{b})} \in \mathbb{R}^{k! \times k! \times k! \times k!} \) on the bottom boundary are given by \begin{align}
    \raisebox{-38pt}{\includegraphics{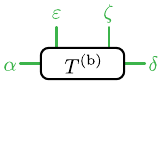}} & = \raisebox{-38pt}{\includegraphics{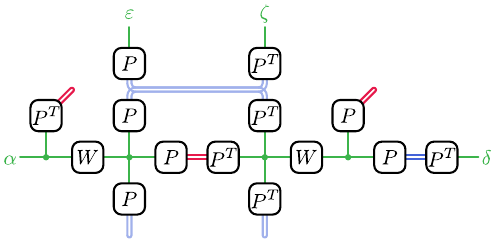}} \\
    & = \sum_{\beta, \gamma, \sigma \in S_k} \Wg \left ( \beta \alpha^{- 1}, d D \right ) \Wg \left ( \sigma \gamma^{- 1}, d D \right ) \nonumber \\
    & \hphantom{{} = \sum_{\beta, \gamma, \sigma \in S_k}} {} \times d^{\# (\alpha)} d^{\# (\sigma)} d^{\# \left ( \beta \gamma^{- 1} \right )} \nonumber \\
    & \hphantom{{} = \sum_{\beta, \gamma, \sigma \in S_k}} {} \times D^{\# \left ( \sigma \delta^{- 1} \right )} D^{\# (\beta) / 2} D^{\# (\gamma) / 2} D^{\# \left ( \beta \zeta^{- 1} \right ) / 2} D^{\# \left ( \varepsilon \gamma^{- 1} \right ) / 2},
\end{align} those of the tensor \( T^{(\mathrm{br})} \in \mathbb{R}^{k! \times k! \times k!} \) in the bottom-right corner are given by \begin{align}
    \raisebox{-38pt}{\includegraphics{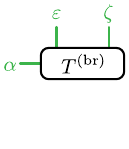}} & = \raisebox{-38pt}{\includegraphics{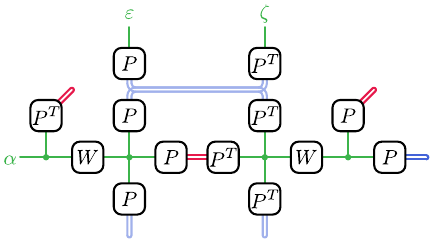}} \\
    & = \sum_{\beta, \gamma, \sigma \in S_k} \Wg \left ( \beta \alpha^{- 1}, d D \right ) \Wg \left ( \sigma \gamma^{- 1}, d D \right ) \nonumber \\
    & \hphantom{{} = \sum_{\beta, \gamma, \sigma \in S_k}} {} \times d^{\# (\alpha)} d^{\# (\sigma)} d^{\# \left ( \beta \gamma^{- 1} \right )} \nonumber \\
    & \hphantom{{} = \sum_{\beta, \gamma, \sigma \in S_k}} {} \times D^{\# (\sigma)} D^{\# (\beta) / 2} D^{\# (\gamma) / 2} D^{\# \left ( \beta \zeta^{- 1} \right ) / 2} D^{\# \left ( \varepsilon \gamma^{- 1} \right ) / 2},
\end{align} those of the tensor \( T^{(\mathrm{r})} \in \mathbb{R}^{k! \times k! \times k! \times k! \times k!} \) on the right boundary are given by \begin{align}
    \raisebox{-28pt}{\includegraphics{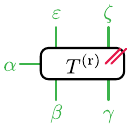}} & = \raisebox{-28pt}{\includegraphics{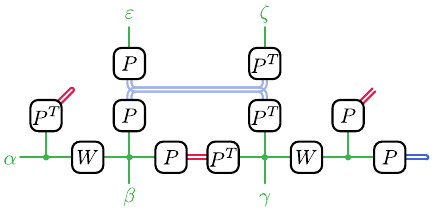}} \\
    & = \sum_{\sigma \in S_k} \Wg \left ( \beta \alpha^{- 1}, d D \right ) \Wg \left ( \sigma \gamma^{- 1}, d D \right ) \nonumber \\
    & \hphantom{{} = \sum_{\sigma \in S_k}} {} \times P_\sigma^{(d)} \nonumber \\
    & \hphantom{{} = \sum_{\sigma \in S_k}} {} \times d^{\# (\alpha)} d^{\# \left ( \beta \gamma^{- 1} \right )} \nonumber \\
    & \hphantom{{} = \sum_{\sigma \in S_k}} {} \times D^{\# (\sigma)} D^{\# \left ( \beta \zeta^{- 1} \right ) / 2} D^{\# \left ( \varepsilon \gamma^{- 1} \right ) / 2} ,
\end{align} those of the tensor \( T^{(\mathrm{tr})} \in \mathbb{R}^{k! \times k! \times k!} \) in the top-right corner are given by \begin{align}
    \raisebox{-28pt}{\includegraphics{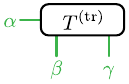}} & = \raisebox{-28pt}{\includegraphics{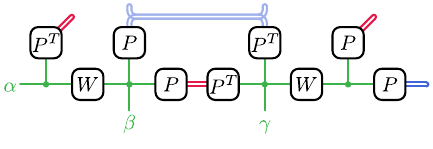}} \\
    & = \sum_{\sigma \in S_k} \Wg \left ( \beta \alpha^{- 1}, d D \right ) \Wg \left ( \sigma \gamma^{- 1}, d D \right ) \nonumber \\
    & \hphantom{{} = \sum_{\sigma \in S_k}} {} \times d^{\# (\alpha)} d^{\# (\sigma)} d^{\# \left ( \beta \gamma^{- 1} \right )} \nonumber \\
    & \hphantom{{} = \sum_{\sigma \in S_k}} {} \times D^{\# ( \sigma)} D^{\# (\beta) / 2} D^{\# (\gamma) / 2},
\end{align} those of the tensor \( T^{(\mathrm{t})} \in \mathbb{R}^{k! \times k! \times k! \times k!} \) on the top boundary are given by \begin{align}
    \raisebox{-28pt}{\includegraphics{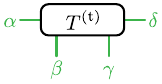}} & = \raisebox{-28pt}{\includegraphics{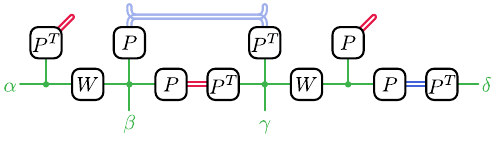}} \\
    & = \sum_{\sigma \in S_k} \Wg \left ( \beta \alpha^{- 1}, d D \right ) \Wg \left ( \sigma \gamma^{- 1}, d D \right ) \nonumber \\
    & \hphantom{{} = \sum_{\sigma \in S_k}} {} \times d^{\# (\alpha)} d^{\# (\sigma)} d^{\# \left ( \beta \gamma^{- 1} \right )} \nonumber \\
    & \hphantom{{} = \sum_{\sigma \in S_k}} {} \times D^{\# \left ( \sigma \delta^{- 1} \right )} D^{\# (\beta) / 2} D^{\# (\gamma) / 2},
\end{align} and those of the tensor \( T^{(\mathrm{t} \ell)} \in \mathbb{R}^{k! \times k! \times k!} \) in the top-right corner are given by \begin{align}
    \raisebox{-28pt}{\includegraphics{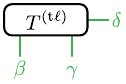}} & = \raisebox{-28pt}{\includegraphics{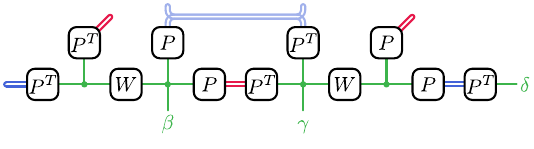}} \\
    & = \sum_{\alpha, \sigma \in S_k} \Wg \left ( \beta \alpha^{- 1}, d D \right ) \Wg \left ( \sigma \gamma^{- 1}, d D \right ) \nonumber \\
    & \hphantom{{} = \sum_{\alpha, \sigma \in S_k}} {} \times d^{\# (\alpha)} d^{\# (\sigma)} d^{\# \left ( \beta \gamma^{- 1} \right )} \nonumber \\
    & \hphantom{{} = \sum_{\alpha, \sigma \in S_k}} {} \times D^{\# (\alpha)} D^{\# \left ( \sigma \delta^{- 1} \right )} D^{\# (\beta) / 2} D^{\# (\gamma) / 2}.
\end{align}

\subsubsection{\texorpdfstring{Case of \( r = 4 t + 1 \)}{Special case}}

As stated in Eq.~\eqref{eq:peps_special_average}, for the case of odd \( x_1 \), even \( x_2 \), \( r_1 = 5 \), \( r_2 = 0 \) and \( t = 1 \), the average of the \( k \)th moment of \( D_1^{\alpha \beta} (x, r, t) \) is given by \begin{align}
    D_1^{\alpha \beta} (x, r, t) = \frac{1}{d^6 D^7} \vcenter{\hbox{\includegraphics{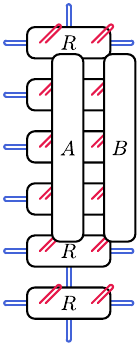}}} = \frac{1}{d^6 D^7} \vcenter{\hbox{\includegraphics{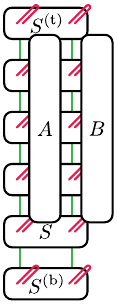}}} = \frac{1}{d^6 D^7} \vcenter{\hbox{\includegraphics{peps_special_average_rhs}}},
\end{align} where we have cut permutation-valued legs instead of bond legs in the second step.

Using Eq.~\eqref{eq:permutation_contraction}, the entries of the tensor \( T \in \mathbb{R}^{\times k! \times k! \times k! \times k!} \) in the bulk are given by \begin{align}
    \raisebox{-28pt}{\includegraphics{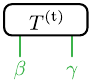}} & = \raisebox{-28pt}{\includegraphics{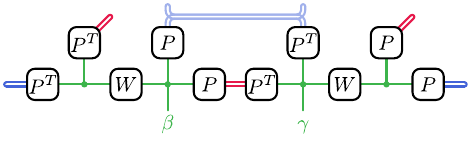}} \\
    & = \sum_{\alpha, \sigma \in S_k} \Wg \left ( \beta \alpha^{- 1}, d D \right ) \Wg \left ( \sigma \gamma^{- 1}, d D \right ) \nonumber \\
    & \hphantom{{} = \sum_{\alpha, \sigma \in S_k}} {} \times d^{\# (\alpha)} d^{\# (\sigma)} d^{\# \left ( \beta \gamma^{- 1} \right )} \nonumber \\
    & \hphantom{{} = \sum_{\alpha, \sigma \in S_k}} {} \times D^{\# (\alpha)} D^{\# (\sigma)} D^{\# (\beta) / 2} D^{\# (\gamma) / 2},
\end{align} those of the tensor \( T^{(\mathrm{t})} \in \mathbb{R}^{k! \times k!} \) on the top boundary are given by \begin{align}
    \raisebox{-28pt}{\includegraphics{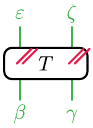}} & = \raisebox{-28pt}{\includegraphics{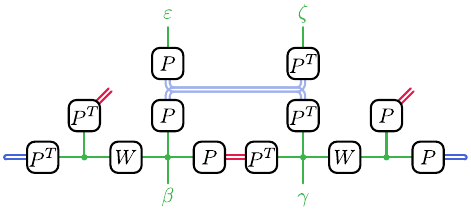}} \\
    & = \sum_{\alpha, \sigma \in S_k} \Wg \left ( \beta \alpha^{- 1}, d D \right ) \Wg \left ( \sigma \gamma^{- 1}, d D \right ) \nonumber \\
    & \hphantom{{} = \sum_{\alpha, \sigma \in S_k}} {} \times \left ( P_\alpha^{(d)} \right )^T P_\sigma^{(d)} \nonumber \\
    & \hphantom{{} = \sum_{\alpha, \sigma \in S_k}} {} \times d^{\# \left ( \beta \gamma^{- 1} \right )} \nonumber \\
    & \hphantom{{} = \sum_{\alpha, \sigma \in S_k}} {} \times D^{\# (\alpha)} D^{\# (\sigma)} D^{\# \left ( \beta \zeta^{- 1} \right ) / 2} D^{\# \left ( \varepsilon \gamma^{- 1} \right ) / 2},
\end{align} and those of the tensor \( T^{(\mathrm{b})} \in \mathbb{R}^{\times k! \times k!} \) on the bottom boundary are given by \begin{align}
    \raisebox{-38pt}{\includegraphics{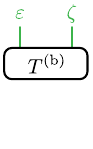}} & = \raisebox{-38pt}{\includegraphics{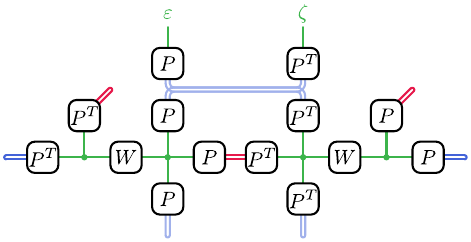}} \\
    & = \sum_{\alpha, \sigma \in S_k} \Wg \left ( \beta \alpha^{- 1}, d D \right ) \Wg \left ( \sigma \gamma^{- 1}, d D \right ) \nonumber \\
    & \hphantom{{} = \sum_{\alpha, \sigma \in S_k}} {} \times d^{\# (\alpha)} d^{\# (\sigma)} d^{\# \left ( \beta \gamma^{- 1} \right )} \nonumber \\
    & \hphantom{{} = \sum_{\alpha, \sigma \in S_k}} {} \times D^{\# (\alpha)} D^{\# (\sigma)} D^{\# (\beta) / 2} D^{\# (\gamma) / 2} D^{\# \left ( \beta \zeta^{- 1} \right ) / 2} D^{\# \left ( \varepsilon \gamma^{- 1} \right ) / 2}.
\end{align}

\subsection{Proof of Corollary~\ref{cor:peps_average_correlation_function}} \label{app:peps_average_correlation_function}

\pepsaveragecorrelationfunction*

\begin{proof}
    The proof is identical to that of Corollary~\ref{cor:mps_average_correlation_function}. That is, for \( k = 1 \), the Weingarten matrix \( W = G^{- 1} \) (see Sec.~\ref{sec:k-twirl}) has only one entry, namely \begin{align}
        \Wg (e, d D) = \frac{1}{d D}.
    \end{align} With that, one may confirm that \begin{align}
        \mathbb{E} D_1^{\alpha \beta} (x, r, t) = \mathbb{E} D_2^{\alpha \beta} (x, r, t) = \frac{1}{d^2} \raisebox{-50pt}{\includegraphics{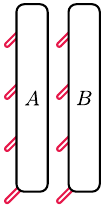}} = \frac{1}{d^2} \tr (A) \tr (B).
    \end{align} Using the spatial unitarity of \( \mathcal{U} \), it holds that \begin{align}
        & \tr (A) = \tr \left ( a^\alpha \right ), \\
        & \tr (B) = \tr \left ( a^\beta \right ).
    \end{align} The statement follows.
\end{proof}

\end{document}